\documentclass[a4paper]{article}

\usepackage[margin=1in]{geometry} 

\usepackage{amsmath}
\usepackage{amsthm}
\usepackage{amssymb}
\usepackage{mdframed}
\usepackage{amsbsy}

\usepackage[utf8]{inputenc}
\usepackage{graphicx}

\usepackage{hyperref}

\hypersetup{
    colorlinks=True,
    linkcolor=blue,
    citecolor=cyan,
    pdfborder={0 0 0},   
    pdfauthor={Author One, Author Two, Author Three},
    pdftitle={A simple article template},
    pdfsubject={A simple article template},
    pdfkeywords={article, template, simple},
    pdfproducer={LaTeX},
    pdfcreator={pdflatex}
}

\renewcommand{\vec}[1]{{\bf #1}}

\usepackage{ulem}

\usepackage[sort&compress,numbers,square]{natbib}

\newcommand{\be}{\begin{equation}}
\newcommand{\ee}{\end{equation}}
\newcommand{\<}{\langle}
\renewcommand{\>}{\rangle}
\newcommand{\Tr}{{\rm Tr\,}}

\newtheorem{theoremS}{Theorem}

\newtheorem{thrm}{Theorem}
\newtheorem{prop}{Proposition}
\newtheorem{cor}{Corollary}
\newtheorem{lem}{Lemma}

\newtheorem{dfn}{Definition}
\newtheorem{notn}{Notation}

\newcounter{mycounter}
\date{} 

\title{Complexity of Local Quantum Circuits\\ under Nonunital Noise}
\author{ 
Oles Shtanko\thanks{IBM Quantum, IBM Research – Almaden, San Jose CA, 95120, USA (oles.shtanko@ibm.com)}, \;
Kunal Sharma\thanks{IBM Quantum, IBM Thomas J. Watson Research Center, Yorktown Heights, NY 10598, USA (kunals@ibm.com)}}

\begin{document}

\maketitle

\begin{abstract}
It is widely accepted that noisy quantum devices are limited to logarithmic depth circuits unless mid-circuit measurements and error correction are employed. However, this conclusion holds only for unital error channels, such as depolarizing noise. Building on the idea of the ``quantum refrigerator'' \cite{ben2013quantum}, we improve upon previous results and show that geometrically local circuits in the presence of nonunital noise, in any dimension $d\geq 1$, can correct errors without mid-circuit measurements
and extend computation to any depth, with only polylogarithmic overhead in the depth and the number of qubits. This implies that local quantum dynamics subjected to sufficiently weak nonunital noise is computationally  universal and nearly as hard to simulate as noiseless dynamics. Additionally, we quantify the contraction  property of local random circuits in the presence of nonunital noise. 
\end{abstract}

\vspace{0.4cm}

\section{Introduction}

Noise significantly affects the performance of quantum algorithms \cite{cai2024shor} and is one of the most critical challenges for every quantum experiment today \cite{preskill2018quantum}. Despite these obstacles, there has been a recent surge in noisy quantum hardware, which has demonstrated the ability to generate results in the high-complexity regime, approaching the frontier of classical simulation techniques \cite{kim2023evidence, robledo2024chemistry}. An interesting open question is whether noisy systems can provide an advantage as the coherence times keep improving and the number of qubits scales up, before quantum error correction becomes available. Interestingly, the resolution of this question seems to depend on the nature of the noise being considered.

Unital noise models, such as depolarizing noise, are widely used in the field due to their simplicity: they can be simulated by summing stochastic unitary trajectories. Their simple representation not only makes it easy to analyze noisy quantum circuits analytically \cite{aharonov1996limitations, aharonov2023polynomial,deshpande2022tight,fontana2023classical, schuster2024polynomial} but has also recently enabled the successful implementation of error mitigation techniques on current noisy quantum hardware \cite{van2023probabilistic,kim2023evidence}.
However, such noise rarely reflects the actual physics of quantum devices \cite{pino2021demonstration,krantz2019quantum}. Furthermore, this noise generally leads to an irreversible increase in the entropy of the system until it reaches the trivial maximally mixed state. As a result, the output of such noisy random circuits approaches a uniform distribution exponentially with increasing circuit depth \cite{aharonov1996limitations,pastawski2009how, gao2018efficient,aharonov2023polynomial}, rendering circuits with superlogarithmic depth classically simulable. 

The effect of more general, \textit{nonunital} noise is less studied. In \cite{mele2024noise}, it was shown that, similar to unital noise, the output distribution of any circuit subject to general noise converges at logarithmic depth—but only for noise strength values exceeding a certain constant value. For sufficiently weak noise, however, this convergence result does not hold. The seminal work \cite{ben2013quantum} provided strong evidence that circuits with all-to-all connectivity, when subjected to strictly nonunital noise, can efficiently simulate noiseless circuits. Here, we generalize this result to arbitrary local quantum circuits in dimensions $d\geq 1$.

Circuits exhibiting such complexity must be carefully constructed to exploit properties of noise. One might ask whether a similar result holds for \textit{random} circuits. As it was shown in \cite{fefferman2023effect}, anti-concentration—typically observed in the noiseless case \cite{dalzell2022random}—does not occur when nonunital noise is present, which implies that current efficient classical algorithms for sampling from noisy circuits do not directly apply to this case  
 \cite{aharonov2023polynomial}.  However, as the circuit depth increases, the distance between any two input states inevitably collapses exponentially with the number of qubits after a certain depth, even under nonunital noise \cite{mele2024noise}. In this work, we demonstrate that for general nonunital noise, this depth is at least logarithmic for local architectures, while for $n$-qubit Haar random operations it is constant. 

\subsection{Overview of the result}

In our work, we address two main questions. First, we ask: does the result of Ref.~\cite{ben2013quantum} generalize to arbitrary architectures (including 1D) with the same polylogarithmic overhead? We answer this question positively in Theorem~\ref{thm:main_main} below. We show that this is possible if the noise strength is below a certain critical value that depends on the introduced measure of noise nonunitality and the dimension of the lattice. This result establishes that the complexity of nonunital noise is essentially independent of architecture and can be applied to all existing hardware architectures.

To prove this result, we employ the same quantum refrigerator construction as in Ref.~\cite{ben2013quantum}, but we incorporate the locality structure of the circuit and account for errors occurring during qubit swaps. We also derive how the resulting noise threshold depends on the lattice dimension. In the limit of all-to-all connectivity, we restore the original result from \cite{ben2013quantum}.

In the second part, we examine random circuits affected by general types of noise. In Theorem~\ref{thm_low_bound}, we demonstrate that the distance between two initially orthogonal product states in the computational basis contracts at a rate bounded above by a constant, for local circuits. This result extends the results of Ref.~\cite{deshpande2022tight} to nonunital noise models and complements the results of \cite{mele2024noise}. 
Moreover, in Theorem~\ref{thrm:all_to_all_gates} we provide an example of noisy all-to-all Haar random circuit, in which the distance between two states becomes exponentially small in the number of qubits already at a constant depth.

\subsection{Worst-case complexity}

We define a $d$-dimensional local noisy quantum circuit with noise generated by the map $\mathcal{N}$ if it can be represented as a sequence of $m$-qubit unitary operations applied between the neighboring nodes of a $d$-dimensional lattice, each followed by a product of single-qubit CPTP maps $\mathcal{N}^{\otimes m}$ acting on the same $m$ qubits. We assume that the idling qubits are affected by the same noise channel $\mathcal N$. This work focuses on general, nonunital noise. Following the common definition, we define the noise as nonunital if it does not preserve the identity operator.
\begin{dfn}\label{def:nonunital_channel} We define a channel $\mathcal N$ to be nonunital if $\|\mathcal N(I)-I\|_1>0$, where $I$ is the identity operator.
\end{dfn}
Common examples of nonunital channels considered here include the replacement channel (see Definition~\ref{dfn:rep_channel}) and the amplitude damping channel (see Definition~\ref{dfn:gen_damping}). Generally, the noise strength of a channel $\mathcal{N}$ is defined as 
\begin{align}\label{eq:noise-strength}
\kappa = \|\mathcal{N} - \mathcal{I}\|_{\diamond}~,    
\end{align} 
where $\mathcal{I}$ denotes the identity channel. Additionally, we assume that, for a typical channel, the distance between any pair of states contracts at a rate proportional to $\kappa$. Let $\mathsf{S}_N = \{\rho: \rho \geq 0, \Tr \rho = 1, {\rm dim}(\rho) = N\}$ denote a set of density matrices on an $N$-dimensional Hilbert space. Then the contraction condition for $\mathcal{N}$  can be formulated as follows:
\begin{dfn}\label{def:contracting}
We define a channel to be contracting if for any $\rho, \sigma \in \mathsf S_2$, there exists $\Delta > 0$ such that
\be\label{eq:assump_contr}
\|\mathcal N(\rho - \sigma)\|_1 \leq (1 - \Delta \kappa)\|\rho - \sigma\|_1,
\ee
where $\kappa = \|\mathcal{N} - \mathcal{I}\|_{\diamond}$.
\end{dfn}

This condition applies to most quantum channels, except for a small number of special cases, such as unitary channels and dephasing noise $\mathcal{N}_{\rm deph}(\rho) := \frac{\gamma}{2} Z\rho Z + \left(1 - \frac{\gamma}{2}\right)\rho$, where $\gamma \in [0, 1]$ and $Z$ is the Pauli-Z operator, which preserves the distance between any two states in the computational basis. 

Finally, let $\sigma^* \in \mathsf{S}_2$ to be a fixed point of the channel $\mathcal{N}$, defined by $\mathcal{N}(\sigma^*) = \sigma^*$. The fixed point can also be characterized by the purity parameter defined as 
\begin{align}\label{eq:eta}
\eta := \sqrt{2\Tr[(\sigma^*)^2] - 1}~.
\end{align}

From Definition~\ref{def:contracting}, it follows that $\mathcal N$ must have a unique fixed point. Moreover, the parameter $\eta$ for the fixed point of a nonunital channel must satisfy $\eta > 0$ since $\eta = 0$ implies a maximally mixed state which is excluded by Definition~\ref{def:nonunital_channel}.
Thus, the parameter $\eta$ serves as an alternative measure of the nonunitality of a channel. 
Our goal is to show that for any $\eta>0$ there exists sufficiently small $\kappa(\eta)$ so that simulating the resulting noisy dynamics is nearly as hard as simulating the noiseless dynamics. In particular, we prove the following result.

\begin{figure}[t!]
    \centering
    \includegraphics[width=0.6\textwidth]{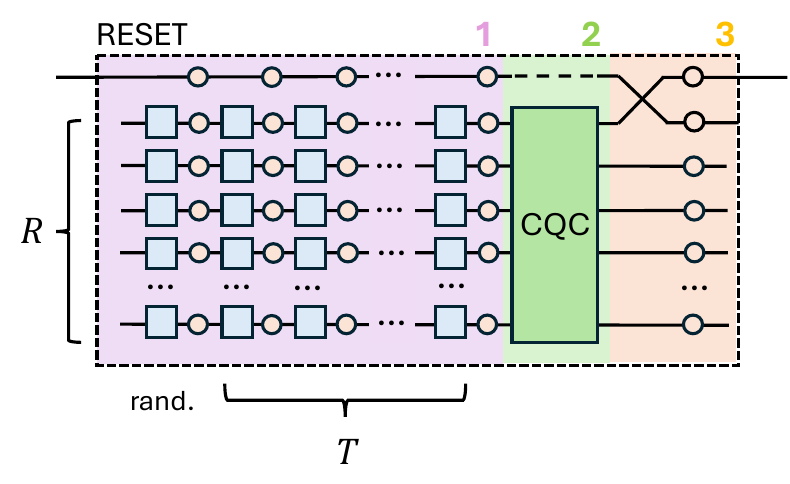}
    \caption{\textbf{RESET circuit.} Blue boxes show the single qubit gates, circles show the noise channel $\mathcal N(\cdot)$, green box shows the Compound Quantum Compressor (CQC) circuit. The RESET is implemented in three steps, marked by different colors.}
    \label{fig:reset}
\end{figure}

\begin{thrm}[Informal]\label{thm:main_main}
Consider $d$-dimensional quantum circuits with noise generated by a contracting noise channel $\mathcal{N}$ with strength $\kappa$ and a unique fixed point $\sigma^*$ with purity parameter $\eta$. Then, for an arbitrary $\eta > 0$,  there exists nonzero $\kappa = O(\eta^{\mu/d})$ with $\mu > \mu_0 = \log 3 / \log(3/2)$, such noisy circuits allow for the simulation of noiseless unitary circuits with a logarithmic overhead.
\end{thrm}

The formal statement, along with a proof, can be found in Section~\ref{sec:main_result}. Let us outline the steps of the proof. First, we need to show that using nonunital noise and ancilla qubits, it is possible to generate an approximate RESET that returns the target qubit to the state $|0\>$. To emulate RESETs, we use a modified Quantum Refrigerator (QR) algorithm originally introduced in Ref.~\cite{ben2013quantum}. Our algorithm consists of three independent steps (see Fig.~\ref{fig:reset}).

In the first step, we first erase the ancilla qubits by applying random unitary gates. These gates should be chosen independently at random for each instance of the circuit. Therefore, after the first step, the ancilla qubits approach the maximally mixed state. This step ensures the absence of correlated errors when the ancilla are reused for quantum error correction. Then, we apply a sequence of single-qubit gates followed by the noise channels until the ancilla qubits reach a certain nonzero polarization. In particular, we show that for nonunital noise, one can construct a single-qubit noisy unitary circuit of depth $T$ such that the output state $\rho_T$ satisfies the following equation  (see Lemma~\ref{lem:settling})
\be\label{eq:colleN_purity32}
z_T :=\Tr(Z\rho_T) = \Omega(\kappa\eta T),
\ee
where $Z$ is Pauli-Z matrix. These ``settled'' ancilla qubits enter the second stage, called compression. In this stage, we use algorithmic cooling  to produce fewer ancilla qubits with higher polarization \cite{fernandez2004algorithmic}. Given $N_a$ ancilla qubits with polarization $z_{\rm in}$ as in Eq.~\eqref{eq:colleN_purity32}, we get the following condition on $N_a$ to obtain a single ancilla with polarization $z_{\rm out}$   using the compression step (see Lemma~\ref{lem:noisy_lemma}),
\be\label{eq:num_anc_intr}
N_a \leq C \frac{\left(-\log 3[1-z_{\rm out}-O(g)]\right)^{\alpha}}{z_{\rm in}^{\mu+O(g)}}~,
\ee
where $\mu = \log 3/\log(3/2) + \delta$ for $\delta>0$, $C$ is a constant, and $g$ is a small parameter chosen such that $\kappa N_a^{1/d}\leq g$. Vanishing $g$ implies that the noise introduced by swapping of qubits within the ancilla blocks is negligible.  
Combining the results in Eq.~\eqref{eq:colleN_purity32} and \eqref{eq:num_anc_intr}, we choose the settling depth
\be
T \propto (\kappa\eta)^{-\mu/(\mu+d)},
\ee
where $\eta$ is the purity of the fixed point of $\mathcal{N}$. Then, we obtain the required number of ancilla per RESET operation to be (see Lemma~\ref{lem:cqc_circ})
\be\label{eq:ancilla_number}
N_{a} =O\left(\frac{\log^{\alpha}(\kappa\eta^{-\nu/d}) }{(\kappa\eta)^{\nu d/(\nu+d)}}\right).
\ee

In the last step, we exchange the output of the compression cycle with the target qubit. Combining the idle error for the logical qubits and swap times, we obtain the noise cost of producing a RESET. In particular, the new effective noise channel per layer $\mathcal N'$ will have the rate $\kappa' = \|\mathcal N'-\mathcal I\|_{\diamond}$ described by the expression (see Corollary~\ref{thm:reset})
\be\label{eq:effective_noise}
\kappa' = \tilde O\left((\kappa\eta^{-\mu/d})^{d/(\mu+d)}\right).
\ee
While this stregnth is larger than the original strength $\kappa$, the ability to implement RESETs allows us to implement quantum error correction and thus reduce the noise strength to an arbitrarily small value. Of course, this is only possible if the original noise strength satisfies the asymptotic behavior $\kappa = O(\eta^{\mu/d})$, otherwise, the expression inside big-O scaling is not bounded.

In particular, we show that for sufficiently small $\kappa'$ in Eq.~\eqref{eq:effective_noise}, local noisy circuits with RESET operations can implement fault-tolerant quantum computing with logarithmic overhead in the number of qubits \cite{gottesman2000fault,aliferis2005quantum,boykin2010algorithms,heussen2024measurement}. Thus, by adding effective RESETs, a noisy circuit of depth $D$ on $N$ qubits can be used to simulate a noiseless circuit of depth $D' < D$ on $N' < N$ qubits. The overhead for this simulation is the product of the overheads required for implementing RESETs and for quantum error correction. The latter can be estimated as $\log (N'D'/\epsilon)$, where $\epsilon$ is the error in simulating a unitary circuit (see Theorem~\ref{thm:threshold}). The former, which corresponds to the number of ancillas needed per logical qubit, is simply given by $N_a$ in Eq.~\eqref{eq:ancilla_number}.

 It is important to note that, unlike the case described in Ref.~\cite{ben2013quantum}, the $\kappa = O(\eta^{\mu/d})$ threshold depends on the purity $\eta$ of the fixed point of the noise and the dimensionality of the lattice $d$. This is due to the fact that, when noise is too strong, errors accumulate during the swaps that connect logical qubits across different ancilla blocks faster than the ancillas can cool under the same noise channel. However, below the noise threshold, the resulting overhead for implementing RESETs remains constant, and the scheme works even for $d=1$. This approach contrasts with the local scheme proposed in Ref.~\cite{ben2013quantum}, which requires polynomial overhead and only works for systems in $d \geq 2$.

\subsection{Average-case results}

In the previous section, we demonstrated that local circuits subjected to nonunital noise remain as powerful as noiseless circuits, even when the degree of nonunitality is arbitrarily small. The goal of this section is to explore a scenario where the gates are  generated randomly. This result can be seen as a generalization of Ref.~\cite{deshpande2022tight} applied to unital noise and complements Ref.~\cite{mele2024noise} which establishes similar upper bounds.

To define noisy parallel random circuits, we first consider an architecture which is defined as an arrangement of gates without specifying the exact gate operations. We focus on architectures that are parallel, meaning the architecture can be divided into parallel layers such that each qubit is involved in exactly one gate per layer. The number of such layers is referred to as the parallel depth, denoted by $D$. In a noisy architecture, it is further assumed that noise is homogeneous, i.e. each gate is followed by the same two-qubit quantum noise channel. Unlike in the previous section, we allow the noise to be a more general form of a two-qubit channel that is symmetric under the swapping of qubits. Our results can be generalized to the case where the noise is neither symmetric nor homogeneous.

A random circuit is defined as a collection of circuits with a specified noisy architecture in which the gate unitaries are distributed according to the Haar measure. We now state the following result: if starting from distinct pure computational basis states, these states remain efficiently distinguishable for logarithmically long times.

\begin{thrm}\label{thm_low_bound}
Consider a noisy random unitary circuit consisting of $D$ parallel layers of two-qubit gates. For any bit strings $\vec{z}, \vec{z}' \in \{0,1\}^n$, with $\vec{z} \neq \vec{z}'$, the corresponding initial quantum states $\rho = |\vec{z}\rangle \langle \vec{z}|$ and $\sigma = |\vec{z}'\rangle \langle \vec{z}'|$ satisfy
\be
\frac{1}{2} \mathbb{E}_{\mathcal B} \|\mathcal C(\rho - \sigma)\|_1 \geq e^{-\Gamma D},
\ee
where $\Gamma>0$ is a constant, $\mathcal{C}$ is a noisy random circuit channel, and $\mathbb{E}_{\mathcal B}$  denotes the expectation over the random two-qubit unitaries.
\end{thrm}

The proof of this theorem can be found in Section~\ref{sec:log_depth_bound}. This result relies on the inherent locality of the gates and does not apply to more general, non-local quantum gates. For example, consider an all-to-all architecture where the system is subjected to a single all-to-all Haar-random unitary gate per layer. For such a circuit, we demonstrate the following result.
\begin{thrm}\label{thrm:all_to_all_gates}
Consider an $n$-qubit all-to-all random circuit of depth $D > 1$ subject to noise generated by a replacement channel with the rate parameter $0 < \gamma \leq 1$. Then
\be
\frac{1}{2} \mathbb{E} \|\mathcal{C}(\rho - \sigma)\|_1 \leq O\left(2^{n/2} \left(1 - \frac{\gamma}{2}\right)^{n(D - 1)/2}\right),
\ee
where $\mathcal{C}$ denotes the channel corresponding to the noisy random circuit, and $\mathbb{E}$ denotes the expectation over the gate unitaries.
\end{thrm}
The proof of this theorem can be found in Section~\ref{sec:const_depth_rc}. This result demonstrates that for all-to-all circuits, the output depends only on the last \textit{constant} number of layers, independent of the noise model or its strength. This result exhibits a tighter scaling than existing bounds for nonunital noise in local circuits \cite{mele2024noise}, although it is derived specifically for the replacement channel. We show that this result can also be extended to the general amplitude damping.

\subsection{Discussion}

In this work, we proved that quantum circuits subjected to nonunital noise are as hard to simulate classically as noiseless dynamics. This result is based on a construction that allows implementing RESET operations using ancilla qubits while maintaining coherence in the remaining qubits. Our results go beyond those in \cite{ben2013quantum}, which required either all-to-all connectivity or a polynomial overhead for local circuits. 
Moreover, our results can be extended to continuous-time dynamics using techniques developed in~\cite{trivedi2022transitions}. 

Our results have several implications. First, the claim that noisy quantum computers are only useful for logarithmic depth \cite{preskill2018quantum} may be misleading, as it relies on the intuition derived from depolarizing and other types of unital noise. One can envision a broad class of circuits in which nonunital dissipation processes are integrated into the design to provide a near-term quantum advantage.
However, it remains unclear whether these circuits can be efficiently learned.

Another implication is the new evidence suggesting that a universal classical algorithm is unlikely to simulate natural dissipative systems interacting with their environment (e.g., a  thermal bath). This underscores the importance of quantum computing for studying states like Gibbs states \cite{chen2023efficient}, which characterize thermal equilibrium in many-body physics.

It is worth noting that the noise threshold obtained in our work using a concatenated error correction scheme is quite low. Thus, an open question remains as to whether it is possible to construct circuits with higher tolerance to nonunital noise. Nonetheless, our result reveals an intriguing behavior: even noise that is $\eta$-close to depolarizing noise allows fault-tolerant quantum computation with local circuits without measurements and feedback, provided $\eta$ is nonzero.

\setcounter{thrm}{0}

\section{Worst-case complexity}

In this section, we present a step-by-step procedure to construct a local \textit{noisy} unitary circuit capable of simulating arbitrary local \textit{noiseless} unitary circuits with polylogarithmic overhead, assuming the noise is strictly nonunital. First, we describe a scheme for a quantum refrigerator, a subcircuit that takes $R$ qubits in a fully mixed state and returns one of them in an \textit{almost} pure state using a sequence of algorithmic cooling procedures \cite{fernandez2004algorithmic}. Using this construction, we show that noisy unitary circuits can simulate other noisy unitary circuits with added RESETs. Finally, we connect the output of such RESET-enabled unitary circuits to the output of noiseless circuits using the threshold theorem \cite{aharonov1997fault,gottesman2000fault}.

\subsection{Noise-enabled RESETs}
\label{sec:reset_explained}

We define a noise-enabled RESET operation as an algorithm that resets a target qubit to the $|0\>$ state using noisy unitary operations only. Such an algorithm operates in three stages: (1) a settling stage, where a set of ancilla qubits is randomized into a maximally mixed state and then exposed to nonunital noise for some depth $T$; (2) a cooling stage, which takes $R$ qubits and returns one qubit in a (nearly) pure state by pushing the entropy to the remaining ancilla qubits; and (3) a stage that swaps the target qubit with the output from the compression stage. Below, we begin our discussion with the settling process in (1), followed by the quantum compressor in (2).

\subsubsection{Settling process}
\label{sec:cooling_process}

We define a $d$-dimensional local noisy quantum circuit generated by the map $\mathcal{N}$ if it can be represented as a sequence of $m$-qubits unitary operations applied between the nodes of a $d$-dimensional lattice, each followed by the product single-qubit CPTP map $\mathcal{N}^{\otimes m}$ acting on the same $m$ qubits. The parameter $\kappa = \|\mathcal{N} - \mathcal{I}\|_{\diamond}$ characterizes the strength of the noise, where $\mathcal{I}$ denotes the single-qubit identity channel. Additionally, we assume that the noise channel $\mathcal{N}$ has a fixed point, denoted by $\sigma^*$, with a fixed-point purity defined as $\eta := \sqrt{2\Tr[\sigma^*]^2 - 1} > 0$.

\begin{dfn}\label{def:diag_chan} Consider a channel $\mathcal N$ with unique fixed point $\sigma^*$, then we define its diagonalized form as 
\begin{align}\label{eq:N'}
    \mathcal N'= \mathcal U\mathcal N \mathcal U^\dag ~,
\end{align} 
where $\mathcal U$ is a unitary that diagonalizes $\sigma^*$, so that the diagonal state $\tilde \sigma^* = \mathcal U(\sigma^*)$ is the fixed point of $\mathcal N'$.
\end{dfn}

Note that $\tilde \sigma^*$ has the same purity as the state $\sigma^*$.  
It is worth noting that channels in a diagonalized form do not necessarily yield a diagonal density matrix when acting on a diagonal input state. For instance, consider the channel $\mathcal{N}' = \sum_{i=1,2} K_i \rho K_i^\dagger$, where $K_1 = \frac{1}{\sqrt{2}}(|0\rangle + |1\rangle)\langle 0|$ and $K_2 = |1\rangle \langle 1|$. This channel is in a diagonal form with a unique fixed point $\sigma^* = |1\rangle \langle 1|$. However, when it acts on an input state such as $|0\>\<0|$, it generates a superposition in the computational basis.

The performance of the settling process can be described in terms of  \textit{polarizations} of the input state $z_{\rm in}$ and the output state $z_{\rm out}$, defined for certain input state $\rho_{\rm in}$ and the output state $\rho_{\rm out}$ as
\be\label{eqs:inout-sdc0}
z_{\rm in} := \Tr (Z\rho_{\rm in}), \qquad z_{\rm out} := \Tr (Z\rho_{\rm out}).
\ee
The diagonalized form of a nonunital channel with a unique fixed point always results in the convergence of $ z_{\rm out} $ to the value $ \eta $, as stated in the following lemma.

\begin{lem}[Contractivity of polarizations] \label{def:contr_chan} 
Consider a nonunital noise channel $\mathcal N$ with a unique fixed point of purity $\eta > 0$, as defined in Eq.~\eqref{eq:eta}, and a diagonal density matrix $\rho_{\rm in}$. Let $\rho_{\rm out} := \mathcal N'(\rho_{\rm in})$, where $\mathcal N'$ is the diagonalized form of the channel $\mathcal N$, as defined in Def.~\ref{def:diag_chan}. Then
\be\label{eq:monotonical_contr}
z_{\rm out} = \eta + \chi(z_{\rm in} - \eta),
\ee
where $\chi < 1$.
\end{lem}

\begin{proof}
We begin with the explicit form of the state $\rho_{\rm in}$ expressed in terms of the polarization $z_{\rm in}$. Since $\rho_{\rm in} = \hat{\mathcal D}(\rho_{\rm in})$, it must have the form
\be
\rho_{\rm in} = \frac{1}{2}\Bigl(I + z_{\rm in} Z\Bigr).
\ee
For this input state, the output of the quantum channel is
\be\label{eq:a-qwduic}
\rho_{\rm out} := \mathcal N'(\rho_{\rm in}) = \frac{1}{2} \Bigl(\mathcal N'(I) + z_{\rm in} \mathcal N'(Z)\Bigr).
\ee
Next, we consider $\rho^*$ to be the fixed point of the channel $\mathcal N'$, which has purity $\eta$ and is diagonal. This allows us to express it as
\be
\mathcal N'(\rho^*) = \rho^* = \frac{1}{2}\Bigl(I + \eta Z\Bigr).
\ee
Using the linearity of quantum maps, and by solving for $\mathcal N'\left(\frac{1}{2}\Bigl(I + \eta Z\Bigr)\right) = \frac{1}{2}\Bigl(I + \eta Z\Bigr)$, we derive the following expression: 
\be\label{eq:NI-tre123}
\mathcal N'(I)  = I + \eta Z - \eta \mathcal N'(Z).
\ee
Substituting $\mathcal N'(I)$ from this relation into Eq.~\eqref{eq:a-qwduic}, we obtain
\be
\rho_{\rm out} = \frac{1}{2} \Bigl(I + \eta Z + (z_{\rm in} - \eta)\mathcal N'(Z)\Bigr).
\ee
By applying the formula $z_{\rm out} = \Tr(\rho_{\rm out} Z)$, we get the expression
\be\label{eqs:chi_def}
z_{\rm out} = \eta + \chi (z_{\rm in} - \eta), \qquad \chi := \frac{1}{2} \Tr(Z \mathcal N'(Z)).
\ee
The final step is to show that $\chi < 1$. To prove this, assume the contrary, i.e., $\chi = 1$, and demonstrate that it leads to a contradiction. Since $\mathcal N'$ is a CPTP map, we can express
\be
\mathcal N'(Z) = \sum_{P_i \in \{X,Y,Z\}} n_i P_i, \qquad n_i := \frac{1}{2} \Tr[P_i \mathcal N'(Z)], \qquad \sum_i n_i^2 \leq 1.
\ee
If $\chi \equiv n_3 = 1$, it implies that $\mathcal N'(Z) = Z$, and by Eq.~\eqref{eq:NI-tre123}, this results in $\mathcal N'(I) = I$, implying that the noise channel $\mathcal N'$ violates Definition~\ref{def:nonunital_channel}. This contradiction shows that $\chi < 1$ must hold to satisfy the conditions of the lemma. Therefore, by combining $\chi<1$ condition with Eq.~\eqref{eqs:chi_def} we get the desired result in Eq.~\eqref{eq:monotonical_contr}. 

\end{proof}

The result of Lemma~\ref{def:contr_chan} demonstrates that it is possible to design a settling process that guarantees a specific level of polarization at the output. In the following lemma, we outline this process and derive an expression for the output polarization.

\begin{lem}[On output of settling process]\label{lem:settling}
There exists a depth-$T$ quantum circuit subject to general contracting nonunital noise $\mathcal N$ characterized by strength $\kappa$ and purity parameter $\eta$ that takes qubits in maximally mixed state and outputs them in a state with polarization
\be\label{eq:settling_process}
z_{\rm out} \geq \eta \Delta \kappa T + O(\eta \kappa^2 T^2).
\ee
where $\Delta$ is introduced in Defintion~\ref{def:contracting}.
\end{lem}

\begin{proof} 
Let $\rho_0$ denote a single-qubit maximally mixed state. We begin with 
$\rho^{\otimes N}_0$ on $N$ qubits. At the $k$th layer, the state of each qubit undergoes the single-qubit unitary operation $\mathcal U_k$ followed by the noise channel $\mathcal N$ according to the rules of the noisy architecture. We set the unitaries to be $\mathcal U_1 = \mathcal U^\dag \mathcal G_1$ and $\mathcal U_k = \mathcal U^\dag \mathcal G_k \mathcal U$ for $2 \leq k \leq T$, where $\mathcal G_k$ is a unitary transformation that is chosen to diagonalize its input. This step is possible because we know the input of this sequence of gates. 
Here, $\mathcal U$ is the unitary transformation that puts the channel in the diagonalized form (see Definition~\ref{def:diag_chan}). While the unitary $\mathcal G_1$ has no effect if we start from a maximally mixed state, we retain it for the sake of generality. 

This process can be put in a concise mathematical form using the following definition:

\begin{dfn} We define a diagonalization procedure $\rho' = \hat{\mathcal D}(\rho)$ as a state-specific unitary rotation such that $\rho'_{01}=\rho'_{10} = 0$ and $\rho'_{11}\geq \rho'_{00}$. 
\end{dfn}

Note that the diagonalization procedure requires knowledge of the input state and is therefore \textit{not} a linear map. This notation allows us to write the process in a form of the map $\mathcal E^{\otimes N}_T$, where the individual single-qubit maps are expressed as
\be
\mathcal E_T := \mathcal N \mathcal U_T\dots \mathcal N \mathcal U_1 =  \mathcal U^\dag ( \mathcal N' \hat{\mathcal D})^T.
\ee
The channel $\mathcal E_T$ can be de-facto transformed into \be\label{eqs:ia09dcew}
\mathcal E'_T :=  \mathcal U\mathcal E_T = ( \mathcal N' \hat{\mathcal D})^T 
\ee
by implementing an additional unitary $\mathcal U$ ``virtually'' by incorporating it into any gate following the settling procedure. This approach avoids introducing an additional noise channel $\mathcal N$, which would otherwise should be accounted for if the transformation was performed directly. 

Consider the ancilla qubits initially in a maximally mixed state ($\rho_0 =\frac 12 I_{2\times 2}$, $z_0 := \Tr(Z \rho_0) = 0$) and subject to a large number of layers $T \gg 1$, as described above. Consider $|\psi_0\>$ and $|\psi_1\>$ to be the eigenstates of the fixed point $\sigma^*$ of the channel $\mathcal N$. Then we can rewrite $\mathcal U^\dag(Z)=|\psi_0\>\<\psi_0|-|\psi_1\>\<\psi_1|$. Following the Definition~\ref{def:contracting}, the parameter in Eq.~\eqref{eqs:chi_def} satisfies
\be
\begin{split}
\chi := \frac{1}{2} \Tr(Z \mathcal N'(Z)) &= \frac{1}{2} \Tr(Z \mathcal U\mathcal N \mathcal U^\dag (Z)) = \frac{1}{2} \Tr\Bigl(\mathcal U^\dag(Z) \mathcal N \bigl(|\psi_0\>\<\psi_0|-|\psi_1\>\<\psi_1|\bigl)\Bigl)\\ 
&\leq \frac 12\| \mathcal N\bigl(|\psi_0\>\<\psi_0|-|\psi_1\>\<\psi_1|\bigl)\|_1
\leq \frac 12(1 - \kappa \Delta) \||\psi_0\>\<\psi_0|-|\psi_1\>\<\psi_1|\|_1 \leq 1 - \kappa \Delta,
\end{split}
\ee
where, as before, $\kappa = \|\mathcal N - \mathcal I\|_{\diamond}$. 
Next, by Lemma~\ref{def:contr_chan}, the state of the ancilla after $k$ cycles of the settling evolution described by the transformation $\mathcal E'_k$ in Eq.~\eqref{eqs:ia09dcew} has polarization $z_k = \Tr \left[ Z \mathcal E'_k(\rho_k)\right]$, which satisfies the inequality
\be
z_{k+1} \geq \eta + \chi(z_k - \eta).
\ee
This inequality implies that, by solving the recurrence condition for $z_0 = 0$, the polarization after $T$ cycles satisfies
\be
z_T \geq \eta \bigl(1 - \chi^{T}\bigr) \geq \eta \Bigl(1 - (1 - \kappa \Delta)^{T}\Bigr).
\ee
Assuming that the product $\kappa T$ is small, the state of the ancilla after $T$ settling steps is given by
\be\label{eq:z_in_setimate}
z_T = \eta \Delta \kappa T + O(\eta \kappa^2 T^2)~,
\ee
where we used the fact that the diagonalization procedure $\hat {\mathcal D}$ does not decrease the polarization.
\end{proof}

Thus, the settling procedure yields ancilla qubits with a small polarization on the order of $\sim \eta \kappa T$. The next step is to design a circuit that uses these ancilla qubits to produce a single ancilla qubit with higher purity.

\subsubsection{Quantum compressor}
\label{sec:QCp_details}

The compression stage can be realized using a circuit referred to below as a quantum compressor (QCp). In the ideal case, without the presence of noise, this circuit implements a unitary transformation $U_C$ such that when applied to a product of identical single-qubit states $ \rho_{\rm in} $, it operates as
\begin{equation}
\Tr_A \Bigl(U_{C}  
\rho^{\otimes N_{\rm in}}_{\rm in}U_C^\dag\Bigl) = \rho^{\otimes N_{\rm out}}_{\rm out}
,
\end{equation}
where $N_{\rm in}$ is the number of input qubits, $N_{\rm out} < N_{\rm in}$ is the number of output qubits, and $\Tr_A$ denotes the partial trace over the $N_{\rm in} - N_{\rm out}$ ancilla qubits. 

We consider a particular type of QCp characterized by the relation
\be\label{eq:compression_relation}
z_{\rm out} =  y z_{\rm in} + O(z_{\rm in}^2),
\ee
where $z_{\rm in}$ and $z_{\rm out}$ are input and output polarization defined in Eq.~\eqref{eqs:inout-sdc0}, $y > 1$ is the gain parameter, indicating that the QCp outputs a state of higher purity. A key parameter is the ratio of the number of input qubits $N_{\rm in}$ and the output qubits $N_{\rm out}$, denoted as
\be\label{eq:R}
R := \frac{N_{\rm in}}{N_{\rm out}}.
\ee
Based on entropy considerations, this ratio must always satisfy $R \geq y^2$ \cite{fernandez2004algorithmic} and it generally depends on the input and output polarizations. With fixed polarization values, the objective is to design a QCp that minimizes $R$.

For example, a QCp for three input qubits can be realized by implementing a CNOT operation on qubit 2 controlled by qubit 3 and then CSWAP on qubits 1 and 3 controlled by qubit 2 (here, the qubits 1 and 3 are swapped of value of qubit 2 is $|0\>$) \cite{fernandez2004algorithmic}. This procedure is described by the unitary
\be
U_C = \begin{pmatrix}
1 & 0 & 0 & 0 & 0 & 0 & 0 & 0 \\
0 & 0 & 0 & 0 & 1 & 0 & 0 & 0 \\
0 & 0 & 1 & 0 & 0 & 0 & 0 & 0 \\
0 & 1 & 0 & 0 & 0 & 0 & 0 & 0 \\
0 & 0 & 0 & 1 & 0 & 0 & 0 & 0 \\
0 & 0 & 0 & 0 & 0 & 0 & 0 & 1 \\
0 & 0 & 0 & 0 & 0 & 0 & 1 & 0 \\
0 & 0 & 0 & 0 & 0 & 1 & 0 & 0
\end{pmatrix}
\ee
For this three-qubit QCp, if $\rho_{\rm in}$ is chosen in the diagonal (computational) basis, the compression relation in Eq.~\eqref{eq:compression_relation} takes the form
\be\label{eq:3qubit_compr_relation}
z_{\rm out} = \frac 32 z_{\rm in} - \frac 12 z_{\rm in}^3.
\ee
The realization of a larger QCp can be achieved by systematically combining several smaller units, such as the 3-qubit processor described above. To illustrate this, consider a compressor consisting of $ k $ successive cycles. During each cycle, the qubits output from the previous cycle are organized into groups of three. Each of these trios is then compressed by the elementary three-qubit compressor and passed on to the next cycle. The input/output ratio, $R$ as in Eq.~\eqref{eq:R} of such a quantum compressor can be expressed as $ R = 3^k $. We call such an algorithm a \textit{compound quantum compressor} (CQC).

\begin{figure}[t!]
    \centering
    \includegraphics[width=0.9\textwidth]{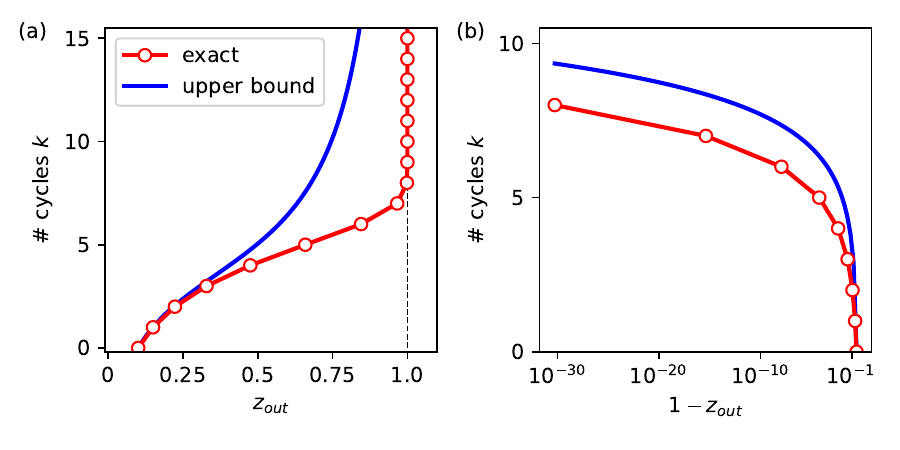}
    \caption{\textbf{Bounds on the number of cycles.} (a) A comparison between the exact number of algorithmic compression cycles (represented by red dots) required to achieve the target polarization $z_{\rm out}$ starting from the initial value $z_{\rm in} = 0.1$ and the upper bound given in Eq.~\eqref{eq:k_small_eta} (dashed line). The bound shows a strong agreement for small values of $z_{\rm out}$. (b) A similar comparison for the exact number of cycles required to achieve $p_{\rm out} = 1 - z_{\rm out}$ starting from the initial value  $p_{\rm in} = 0.3$, alongside the upper bound from Eq.~\eqref{eq:k_small_p}. This bound effectively captures the rapid double-exponential convergence of $p_{\rm out}$ for $p_{\rm in}>1/3$ as the number of cycles increases.}
    \label{fig:noiseless_qd}
\end{figure}

\subsubsection{Compound Quantum Processor (CQC) depth}

The minimum depth required for the CQC to operate depends on the available circuit architecture. Assuming all-to-all connectivity, the depth of such a global CQC, denoted as $N_{\infty}$, is equal to the total depth of all intermediate subcycles. This leads to the expression
\be
N_\infty = N_C \log_3 R,
\ee
where $N_C$ is the depth for one cycle and then we repeat only for $\log_3 R$ cycles. For a spatially local lattice of dimension $d<\infty$, the depth must account for the swaps between the input and output positions of the qubits. We further assume that the qubits form a spatially local subset that fits into a $d$-dimensional cube of linear size $\Theta(R^{1/d})$. In this scenario, each qubit must undergo no more than $\propto R^{1/d}$ swaps to reach the input position for the next cycle. This gives the asymptotic depth of the spatially local CQC as
\be
N_d = N_\infty + \Theta(R^{1/d} \log_3 R) = \Theta(R^{1/d} \log_3 R).
\ee
Note that $R^{1/d}\log_3 R$ comes from the fact that at each step, we need roughly $R^{1/d}\cdot k$ swaps, and we replaced $k = \log_3 R$. 
This depth plays an important role, as the remaining qubits decohere while the compression takes place.

\subsubsection{Input/output ratio for ideal CQC's}\label{sec:noiseless-CQC}

Next, we focus on establishing a bound on $R$, given the polarization of the input state $z_{\rm in}$ and the target polarization of the output state $z_{\rm out}$. We first consider the case of noiseless CQC. In the following section, we extend this result to account for the presence of noise. Additionally, we only consider the compound compressors built from the 3-qubit QCp’s discussed earlier. Such noiseless CQC is characterized by the following Lemma:

\begin{lem}[Noiseless CQC]
There exists a constant $C > 0$, such that for any $0<z_{\rm in}< z_{\rm out}$, $\frac{2}{3} < z_{\rm out} < 1$, and $\delta>0$ there exists a noisy compressor that takes $R$ qubits with polarization $z_{\rm in}$ and returns one of the qubits with polarization $z_{\rm out}$, as defined in Eq.~\eqref{eqs:inout-sdc0}, that satisfies
\be\label{eqs:R_noiseless_sodi}
R \leq C \frac{\left(-\log(3[1 - z_{\rm out}])\right)^{\alpha}}{z_{\rm in}^{\mu_0+\delta}},
\ee
where $\alpha = \log_2 3 \approx 1.59$, $\mu_0 = \frac{\alpha}{\alpha - 1} \approx 2.71$.
\end{lem}

\begin{proof} Consider the CQC described in Section~\ref{sec:QCp_details}, which comprises $k$ cycles of individual 3-qubit QCp layers. Let the input and output polarizations of the QCps for the $m$th cycle be denoted as $z^{(m)}_{\rm in}$ and $z^{(m)}_{\rm out}$, respectively, where $1 \leq m \leq k$. Then, they are related by
\be\label{eqs:zm_rule_xcs}
z^{(m)}_{\rm out} = \frac 32 z^{(m)}_{\rm in} - \frac 12 \left(z^{(m)}_{\rm in}\right)^3,
\ee
which satisfies $z^{(m)}_{\rm out} \geq z^{(m)}_{\rm in}$, indicating that the QCp increases the polarization after each cycle. The output of the $m$-th cycle becomes the input for the $(m+1)$-th cycle, i.e., $z^{(m+1)}_{\rm in} = z^{(m)}_{\rm out}$. These two conditions together imply that $z^{(m+1)}_{\rm out} \geq z^{(m)}_{\rm out}$, leading to non-decreasing polarization across the cycles. Thus, the outputs of the $m$-th cycle can be expressed as
\be
z_{\rm out}^{(m)} =\left(\frac 32-\frac{(z^{(m)}_{\rm in})^2}{2}\right)z^{(m)}_{\rm in} \geq \left(\frac 32-\frac{z_{\rm out}^2}{2}\right)z^{(m)}_{\rm in},
\ee
where we used the property $z^{(m)}_{\rm in}\leq z^{(m)}_{\rm out}\leq z_{\rm out}:=z^{(k)}_{\rm out}$. Using this recurrent relation sequentially, starting with $z^{(1)}_{\rm in} \equiv z_{\rm in}$, the $k$-cycle CQC satisfies
\be\label{eq:inequalities_small_eta}
 z_{\rm out} \geq \left(\frac 32-\frac{z_{\rm out}^2}{2}\right)^k z_{\rm in}.
\ee
 From this expression, the minimum number of cycles $k$ needed to reach the target output polarization satisfies the inequality
\be\label{eq:k_small_eta}
k\leq \frac{\log (z_{\rm out}/z_{\rm in})}{\log (3/2-z_{\rm out}^2/2)}.
\ee
For reference, a comparison between this bound and the exact numerical solution is presented in Fig.~\ref{fig:noiseless_qd}a. This bound is tight if $z_{\text{in}}$ is small and $z_{\text{out}}$ is not too close to unity.

Using Eq.~\eqref{eq:k_small_eta}, we get the following upper bound on $R$
\be\label{eq:small_eta_R}
R = \frac{N_{\rm in}}{N_{\rm out}} = 3^{k} \leq \left(\frac{z_{\rm out}}{z_{\rm in}}\right)^{\mu_0+\delta},\qquad \mu_0 = \frac{\log 3}{\log (3/2)} \approx 2.7095,
\ee
where the correction in the exponent is
\be\label{eqs:delta_param}
\delta := \frac{\log3}{\log(3/2-z_{\rm out}^2/2)}-\frac{\log 3}{\log(3/2)} = O(z_{\rm out}^2).
\ee
As we have seen in Fig.~\ref{fig:noiseless_qd}a, in the scenario where $z_{\text{out}}$ approaches unity, this inequality becomes loose. To target the regime where $1-z_{\text{out}}$ is small, we first introduce depolarization  parameter, defined as
\be
p_{\rm in} := 1-z_{\rm in}, \qquad p_{\rm out} := 1-z_{\rm out}~.
\ee
Then we can rewrite Eq.~\eqref{eqs:zm_rule_xcs} as the update rule
\be\label{eq:noiseless_relations_p}
p^{(m)}_{\rm out} = \left(p^{(m)}_{\rm in}\right)^2(3-2p^{(m)}_{\rm in})\leq 3(p^{(m)}_{\rm in})^2.
\ee
Applying this inequality for the $m$th cycle of the compound compressor and using the fact that $p^{(m+1)}_{\rm in} = p^{(m)}_{\rm out}$, we get the recurrence relation
\be
p^{(m+1)}_{\rm in}\leq 3(p^{(m)}_{\rm in})^2.
\ee
The output of the $k$th cycle satisfies the inequality
\be\label{eq:k_small_p}
p_{\rm out} := p^{(k)}_{\rm out} \leq \frac 13(3p_{\rm in})^{2^k}.
\ee
The value of depolarization decreases monotonically, as long as $p^{(m)}_{\rm in} \leq p_{\rm in} < \frac{1}{3}$. Note that we only invoke this analysis in the third compressor as we describe below. Therefore, we did not require this condition separately in the Lemma statement. Therefore, the number of steps required to achieve the target value of depolarization is bounded as
\be
k \leq \log_2 \left(\frac{\log3\,p_{\rm out}}{\log3\,p_{\rm in}}\right). 
\ee
This bound is also illustrated in Fig.~\ref{fig:noiseless_qd}b, alongside the exact solution. As a result, we obtain the following bound on $R$:
\be\label{eq:small_p_R}
R = 3^{k} \leq \left(\frac{\log 3\,(1-z_{\rm out})}{\log 3\,(1-z_{\rm in})}\right)^\alpha,\qquad \alpha = \log_2 3\approx 1.585,
\ee
Note the logarithmic scaling with $p_{\rm out} = 1-z_{\rm out}$.

We can now combine the two results to derive a hybrid bound on $R$. We consider three compressor steps. The first compressor operates with a small input polarization $z_{\rm in}\ll 1$, producing a constant output polarization of $z_1$ with a compression ratio of $R_1$. The second compressor takes $z_1$ as an input and outputs polarization $z_2>2/3$ with ratio $R_2$, while the third compressor processes qubits with polarization $z_2$ and produces an output polarization of $z_{\rm out} \approx 1$ with ratio $R_3$. For $R_1$ and $R_2$ we use Eq.~\eqref{eq:small_eta_R} and for $R_3$ we use Eq.~\eqref{eq:small_p_R}. Note that $R_2$ is independent of both $z_{\rm in}$ and $z_{\rm out}$, and can be treated as a constant. Consequently, the resulting expression for the input/output ratio of the full compressor satisfies
\be\label{eq:generalization_QCp_performance_bef}
R = R_1R_2R_3 \leq \left(\frac{z_1}{z_{\rm in}}\right)^{\mu_0+\delta} \left(\frac{z_2}{z_{1}}\right)^{\mu_0+\delta'} \left(\frac{\log 3\,(1-z_{\rm out})}{\log 3\,(1-z_2)}\right)^\alpha \leq C \frac{\left(-\log 3(1-z_{\rm out}\right))^\alpha}{z_{\rm in}^{\mu_0+\delta}},
\ee

Here, $C = z_1^{\delta-\delta'}z_2^{\mu_0+\delta'}(-\log(3(1-z_2)))^{-\alpha}$, $\delta = O(z^2_1)$, and $\delta'  = O(z^2_2)$ are independent of both $z_{\rm in}$ and $z_{\rm out}$, and $z_2$ is a constant. By assigning $z_1$ an arbitrarily small constant value, we can achieve arbitrarily small $\delta$, see Eq.~\eqref{eqs:delta_param}, though with an increased value of $C$. This expression completes our proof.
\end{proof} 

The inequality in Eq.~\eqref{eqs:R_noiseless_sodi} highlights an important feature of CQC's: the input-output ratio scales polynomially with the input polarization, but only logarithmically with the output depolarization that constitutes the final RESET error. This logarithmic growth allows the resulting RESET error to be comparable to the noise error on the other qubits, i.e. $1 - z_{\rm out} \sim \kappa$, while allowing the parameter $\kappa R^{1/d}$ to remain small when $\kappa$ is arbitrarily small. As we will see, controlling the strength of this parameter is crucial to the construction of the noisy compressor.

\subsubsection{Input/output ratio for noisy CQC's}

Our next step is to extend the results from the previous section to the noisy CQC. As stated above, we assume that noise is modeled by single-qubit channels $\mathcal{N}$, which act after each two-qubit layer and satisfy the condition $\|\mathcal{N} - \mathcal{I}\|_{\diamond} = \kappa$,
where $\mathcal{I}$ is the single-qubit identity channel, and $\kappa > 0$ represents the noise strength. Recall that we consider $\sigma^*$ as a fixed point of the noise channel $\mathcal{N}$ with a fixed purity parameter $\eta := \sqrt{2\Tr[\sigma^*]^2 - 1} > 0$.

\medskip 

We start by establishing a noisy compression relation for a single noisy QCp. We then establish a bound on the input/output ratio $R$ for a noisy CQC's. 

\begin{lem}[Noisy compression for polarizations]\label{lem:compr_rel_z}
Consider a noisy quantum circuit (QCp) in Eq.~\eqref{eq:3qubit_compr_relation} composed with a constant number of unitary gates subject to noise modeled by a channel $\mathcal N$ with noise strength $\kappa$ and fixed-point purity $\eta$. There exist constants $a, b \geq 0$ such that
\be
z_{\rm out} \geq \left(\frac{3}{2} - a\kappa\right) z_{\rm in} - b\kappa\eta.
\ee
\end{lem}

\begin{proof} Let $\rho$ denote an input state and let $Z_{\rm out}$ denote the Pauli Z operartor on output qubits. Let $\mathcal C$ denote the noisy circuit. Then, the output polarization can be expressed as
\be
z_{\rm out} = \Tr(Z_{\rm out} \mathcal C(\rho)).
\ee
Next, we decompose the initial state as
\be
\rho = \rho_0 + z_{\rm in} \delta\rho_1 + z^2_{\rm in}\delta\rho_2,
\ee
where $\rho_0$ is the maximally mixed state, $\delta\rho_1 = \frac 12 \sum_i Z_i$, and $z^2_{\rm in}\delta \rho_2 : = \rho_0-z_{\rm in}\delta \rho_1$ accounts for the rest. The noisy circuit channel is given by
\be
\mathcal C = \prod_{k=1}^{N_g} \left( \prod_{j \in G_k} \mathcal N_j \right)\mathcal U_k =  \mathcal C_0 + \kappa \delta \mathcal C_1 + \kappa^2 \delta \mathcal C_2,
\ee
where $\mathcal N_j = \mathcal I \otimes \dots \otimes (\mathcal N)_j \otimes \dots \otimes \mathcal I$ is the noise operator acting on the $j$-th qubit, and $N_g$ is the number of gates. The products $\prod_k$ are taken in numerical order. Here, $\mathcal C_0$ represents the noiseless circuit channel, and the correction channel is defined as
\be
\kappa \delta \mathcal C_1 := \sum_{k=1}^{N_g}\sum_{j \in G_k}\prod_{l = k+1}^{N_g}\mathcal U_l \delta \mathcal N_j \prod_{l = 1}^{k}\mathcal U_l, \qquad \delta \mathcal N_j := \mathcal N_j - \hat{\mathcal I},
\ee
with $G_k$ being the list of qubits participating in gate $k$, and $\kappa^2 \delta\mathcal C_2:= \mathcal C-\kappa \delta\mathcal C_1$ accounting for the remaining terms. We denote $\hat{\mathcal I}$ as the identity operator on all qubits in the circuit, distinguishing it from the single-qubit identity operator $\mathcal I$. Combining these two decompositions, we arrive at
\be
\begin{split}
z_{\rm out} = \Tr(Z_{\rm out} \mathcal C_0(\rho)) & + \kappa \Tr(Z_{\rm out} \delta \mathcal C_1(\rho_0)) + z_{\rm in} \kappa \Tr(Z_{\rm out} \delta \mathcal C_1(\delta \rho_1)) \\
& + \kappa z^2_{\rm in} \Tr(Z_{\rm out} \delta \mathcal C_1(\delta \rho_2)) + \kappa^2 \Tr(Z_{\rm out} \delta \mathcal C_2(\rho_0)) \\
& + z_{\rm in} \kappa^2 \Tr(Z_{\rm out} \delta \mathcal C_2(\delta \rho_1)) + z^2_{\rm in} \kappa^2 \Tr(Z_{\rm out} \delta \mathcal C_2(\delta \rho_2)),
\end{split}
\ee
Since $\mathcal C_0$ is a noiseless unitary channel, the first term simplifies to
\be
\Tr(Z_{\rm out} \mathcal C_0(\rho)) = \frac 32 z_{\rm in} - \frac 12 z^3_{\rm in},
\ee
as it does not depend on noise and should match the term from the noiseless circuit expansion, as in Eq.~\eqref{eq:3qubit_compr_relation}. The second term can be rewritten as
\be\label{eqs:proof1psw}
\kappa \Tr(Z_{\rm out} \delta \mathcal C_1(\rho_0)) = \sum_{k=1}^{N_g}\sum_{j \in G_k}\Tr \left( O_k(\mathcal N_j(\hat I) - \hat I) \right),
\ee
where $O_k = \left(\prod_{l = {k+1}}^{N_g} \mathcal U_l \right)^\dag (Z_{\rm out})$ is an operator with eigenvalues $\pm 1$, and $\hat I$ is the identity operator on all $m$ qubits (the hat is again used to distinguish it from the single-qubit operator). Next, we use the fact that the noisy channel $\mathcal N$ has a fixed point $\sigma^*$ with purity parameter $\eta$, implying
\be
\mathcal N(\sigma^*) = \sigma^*, \qquad \sigma^* = \frac 12 I + \frac \eta 2 \left( |\psi_0\>\<\psi_0| - |\psi_1\>\<\psi_1| \right),
\ee
where $|\psi_i\>$ are eigenstates of the operator $\sigma^*$ ordered such that the corresponding eignvalues are sorted in decreasing order. This leads us to
\be
\mathcal N(I) - I = -\left( \eta \mathcal N(|\psi_0\>\<\psi_0|) - \eta |\psi_0\>\<\psi_0| \right) + \eta \mathcal N(|\psi_1\>\<\psi_1|) - \eta |\psi_1\>\<\psi_1|,
\ee
and thus
\be
\begin{split}
\Tr \left( O_k (\mathcal N_j(\hat I) - \hat I) \right) & = \eta \Tr \left( -O_k (\mathcal N_j(\rho^0_j) - \rho^0_j) \right) + \eta \Tr \left( O_k (\mathcal N_j(\rho^1_j) - \rho^1_j) \right),
\end{split}
\ee
where $\rho_j^q := I \otimes \dots \otimes |\psi_q\>\<\psi_q|_j \otimes \dots \otimes I$. By the definition of the diamond norm, we have
\be
\forall O, \|O\| = 1: \qquad \Tr \left( O(\mathcal N_k(\rho^0_k) - \rho^0_k) \right) \geq -\|\mathcal N - \mathcal I\|_{\diamond} = -\kappa.
\ee
This leads to the statement
\be\label{eqs:proof2psw}
\kappa \Tr(Z_{\rm out} \delta \mathcal C_1(\rho_0)) \geq -2N_c \kappa \eta,
\ee
where $N_c$ is the total count of single-qubit noisy channels in the quantum circuit. Similarly to the eqution above, we can show that
\be
\kappa^2 \Tr(Z_{\rm out} \delta \mathcal C_2(\rho_0)) = O(\eta \kappa^2).
\ee
Finally, since the quantum circuit acts on a constant number of qubits and contains a constant number of gates, we find
\be
\begin{split}
& \kappa \Tr(Z_{\rm out} \delta \mathcal C_1(\delta \rho_1)) \sim \kappa \Tr(Z_{\rm out} \delta \mathcal C_1(\delta \rho_2)) = O(\kappa), \\
& \kappa^2 \Tr(Z_{\rm out} \delta \mathcal C_2(\delta \rho_1)) \sim \kappa^2 \Tr(Z_{\rm out} \delta \mathcal C_2(\delta \rho_2)) = O(\kappa^2).
\end{split}
\ee
As a result, we obtain
\be
z_{\rm out} \geq \left(\frac 32 + O(\kappa)\right)z_{\rm in} - 2N_c \kappa \eta + O(\eta \kappa^2) + O(z_{\rm in} \kappa^2).
\ee
Since both $\kappa, z_{\rm out} \leq 1$, we can always find constants $a,b \geq 0$ such that
\be
\begin{split}
&\left(\frac 32 + O(\kappa)\right)z_{\rm in} + O(z_{\rm in} \kappa^2) \geq \left(\frac 32 - a \kappa\right)z_{\rm in}, \\
&-2N_c \kappa \eta + O(\eta \kappa^2) \geq -b \kappa \eta.
\end{split}
\ee
This expression concludes our proof. 
\end{proof}

\begin{lem}[Noisy compression relation for depolarization]\label{lem:compr_rel_p}
Consider a noisy quantum circuit (QCp) in Eq.~\eqref{eq:3qubit_compr_relation} composed with a constant number of unitary gates subject to noise modeled by a channel $\mathcal N$ with strength $\kappa = \|\mathcal N-\mathcal I\|_\diamond$ and fixed-point purity $\eta$. Then there exists a constant $a'> 0$ such that
\be
p_{\rm out} \leq 3p^2_{\rm in} +  a'\kappa.
\ee
where $p_{\rm in/out}:=1-z_{\rm in/out}$ are input/output depolarizations.
\end{lem}
\begin{proof}Let us formally divide the circuit map into two parts:
\be
\mathcal C = \mathcal C_0 + \delta\mathcal C, \qquad \delta \mathcal C := \mathcal C - \mathcal C_0.
\ee
We can then rewrite the output polarization as
\be
p_{\rm out} = \Tr(P_{\rm out} \mathcal C(\rho)) = \Tr(P_{\rm out} \mathcal C_0(\rho)) + \Tr(P_{\rm out} \delta \mathcal C(\rho)).
\ee
Since $\mathcal C_0$ represents the noiseless circuit, we can use the ideal relation
\be
\Tr(P_{\rm out} \mathcal C_0(\rho)) = 3p^2_{\rm in} - 2p^3_{\rm in}.
\ee
Given that the circuit involves a finite number of qubits and gates, we have
\be
\Tr(P_{\rm out} \delta \mathcal C(\rho)) \leq N_c \kappa + O(\kappa^2).
\ee
where, as previous, $N_c$ is the total count of single-qubit noisy channels in the quantum circuit. Thus, the final expression becomes
\be
p_{\rm out} = 3p^2_{\rm in} + N_c \kappa + O(\kappa^2) \leq a'\kappa + 3p^2_{\rm in},
\ee
for some constant $a' > 0$. This completes the proof.
\end{proof}

Next, we use this result to find a bound on the input/output ratio for the compound compressor.

\begin{lem}[Noisy CQC]\label{lem:noisy_lemma}
Consider a single-qubit noise model with strength $\kappa$ and fixed-point purity $\eta$, as in Eqs.~\eqref{eq:noise-strength} and \eqref{eq:eta}, respectively. There exist $C,g_0>0$ such that for any $0<g\leq g_0$, and when $z_{\rm in} = \Omega(\sqrt{\kappa\eta})$, $\frac{11}{12} + O(g) < z_{\rm out}<1-O(g)$ and $\kappa < gR^{-1/d}$, one can construct a noisy compound quantum compressor (CQC), as described in Section~\ref{sec:QCp_details}, with in/out ratio $R$ in Eq.~\eqref{eq:R}, satisfying
\be\label{eqs:noisy_input_output_r2}
R \leq C \frac{\left(-\log 3[1-z_{\rm out}-O(g)]\right)^{\alpha}}{z_{\rm in}^{\mu_0+O(g)}},
\ee
where  $\alpha = \log_23$, $\mu_0 = \log 3/\log(3/2)$.
\end{lem}

\begin{proof} The proof construction is similar to the one used in Section~\ref{sec:noiseless-CQC}.  First, consider a 3-qubit QCp being an element of the $m$th cycle of noisy CQC. We assume that such a compressor has input polarization $z^{(m)}_{\rm in}$ and output polarization $z^{(m)}_{\rm out}$.

Let us apply this Lemma's condition that $z_{\rm in} = \Omega(\sqrt{\kappa\eta})$ and set
\be
z^{(m)}_{\rm in}\geq z_{\rm in}\geq c(\kappa\eta)^{1/2},
\ee
for some constant $c > 0$, where, as before, $z_{\rm in}\equiv z^{(1)}_{\rm in}$. Under this assumption, following the statement of Lemma~\ref{lem:compr_rel_z}, there exist $a$, $b \geq 0$ such that
\be\label{eq:noisy_relatio8}
z^{(m)}_{\rm out} \geq  \left(\frac{3}{2} - a\kappa\right) z^{(m)}_{\rm in}  - b\kappa\eta \geq  \left(\frac{3}{2} - a\kappa\right) z^{(m)}_{\rm in}  - b(z^{(m)}_{\rm in})^2/c^2  \geq \left(\frac{3}{2} - a\kappa - \frac{bz_{\rm out}}{c^2}\right) z^{(m)}_{\rm in},
\ee
where we put $z^{(m)}_{\rm in}\leq z_{\rm out}$ for all $m \leq k$. These inequalities establish the input-output relation for a noisy 3-qubit compressor, with the parameters $a$ and $b$ encoding the details of the circuit architecture. To satisfy the condition $z^{(m)}_{\rm in}\leq z^{(m+1)}_{\rm in}$, we must assume that 
\be
\kappa a + bz_{\rm out}/c^2 \leq 1
\ee
First, from our assumption it follows that  $\kappa \leq g R^{-1/d} \leq g_0 R^{-1/d}$. Then by setting $g_0 < a^{-1}$ and noting that $R \geq 1$, the first term can be made to satisfy $\kappa a <1$. This condition implies that we can operate only for sufficiently small $z_{\rm out}\leq c^2 (1-g_0 a)/b$.

To derive the compound compressor, we must also consider the noise between applications of elementary 3-qubit compressors. The noise accumulates as qubits are swapped between the output of cycle $m$ and the input of cycle $m+1$. Here, after each swap layer, we apply the noisy unitaries $\mathcal N\mathcal U_k$ that we used for the settling process in Section~\ref{sec:cooling_process}. Since the swaps do not generate neither entanglement nor classical correlation between qubits, we can use Lemma~\ref{def:contr_chan} to describe the effect of noise. Since this procedure diagonalizes the input and output states, we have $z_0 = z_{\rm out}^{(m)}$ and $z_T = z_{\rm in}^{(m+1)}$ for our settling process, where $T$ denotes that number of SWAPs needed. Using Lemma~\ref{def:contr_chan} we have the recurrence relation
\be
z_{k+1} \geq \eta + \chi(z_k-\eta) = \eta(1-\chi)+\chi z_k \geq \chi z_k~.
\ee
Since $\eta\geq 0$ and $\chi\leq 1$, we get that
\be\label{z_cool_swaps}
z_T \geq \chi^T z_0.
\ee
In general, a qubit should use no more than $\propto R^{1/d}$ swaps before being positioned as the input for the next layer. Then, after taking into account that $\eta\geq0$, it follows from Eq.~\eqref{z_cool_swaps} that after $T = O(R^{1/d})$ swaps and using Eq.~\eqref{z_cool_swaps} we get
\be\label{eq:next_to_prev32}
z^{(m+1)}_{\rm in} \geq (1-\kappa \Delta)^{O(R^{1/d})}z^{(m)}_{\rm out}\geq (1-f\kappa R^{1/d})z^{(m)}_{\rm out},
\ee
where $f>0$ is a parameter. 
According to the conditions of the Lemma, the value $f\kappa R^{1/d} \leq fg\leq fg_0$. Thus, for $g_0\leq f^{-1}$ we have $f\kappa R^{1/d}\leq 1$ and the bound in Eq.~\eqref{eq:next_to_prev32} is meaningful.

From the inequalities in Eqs.~\eqref{eq:noisy_relatio8}--\eqref{eq:next_to_prev32}, similar to Eq.~\eqref{eq:inequalities_small_eta}, we derive that the number of cycles $k$ required to produce the output with compression parameter $z_{\rm out} \equiv z^{(k)}_{\rm out}$ satisfies
\be
k\leq \frac{\log z_{\rm out}/z_{\rm in}}{\log [(3/2-a\kappa-bz_{\rm out}/c^2)(1-f\kappa R^{1/d})]}.
\ee
Thus we arrive at the noisy compressor in/out ratio in the form
\be\label{eq:noisy_r_bound}
R = \frac{N_{\rm in}}{N_{\rm out}} = 3^{k} \leq \left(\frac{z_{\rm out}}{z_{\rm in}}\right)^{\mu+\delta} \quad {\text{if}}\quad z_{\rm out}\leq c^2 (1-g_0 a)/b,
\ee
where $\mu = \log 3/\log(3/2)$ and $\delta$ is a correction that takes the form
 \be\label{eqs:delta_corr}
\delta = \frac{\log3}{\log[(3/2-a\kappa-bz_{\rm out}/c^2)(1-f\kappa R^{1/d})]}-\frac{\log 3}{\log 3/2} \leq O(z_{\rm out},g).
\ee
Similar to the previous section, the in/out ratio is only tight as a function of $z_{\rm in}$. To improve the bound as a function of $z_{\rm out}$, we consider a similar bound for the parameter $p^{(m)}_{\rm out} := 1-z^{(m)}_{\rm out}$.
According to Lemma~\ref{lem:compr_rel_p}, we get
\be
p^{(m)}_{\rm out} \leq a'\kappa + 3 \left(p^{(m)}_{\rm in}\right)^2.
\ee
In this case, we also need to include the effect of noise during the swaps of the output qubit to position for the next cycle. The noise channel must satisfy 
\be
p^{(m+1)}_{\rm in} \leq p^{(m)}_{\rm out} + f'\kappa R^{1/d},
\ee
for some $f'>0$. The additive term on the right corresponds to the effect of noise after $\propto R^{1/d}$ swaps given that the noise has strength $\kappa$. Using this inequality, we can obtain recurrence relation for the depolarization in the form
\be
p^{(m+1)}_{\rm in} \leq \frac 12 s\kappa R^{1/d} + 3 \left(p^{(m)}_{\rm in}\right)^2,
\ee
where, for simplicity, we defined the parameters $
s := 2(a'+f')$. As the next step, we use the substitution of the variables $ 
\tilde p^{(m)}_{\rm in} := p^{(m)}_{\rm in}-s\kappa R^{1/d}$,
which leads us to the inequality
\be
\tilde p^{(m+1)}_{\rm in} \leq 3 \left(\tilde p^{(m)}_{\rm in}\right)^2 + 3s\kappa R^{1/d} \left(2p^{(m)}_{\rm in}-\frac{1}{6}+s\kappa R^{1/d}\right)
\ee
The rightmost expression is negative and can be omitted without affecting the inequality if make an additional assumption
\be
2 p^{(m)}_{\rm in}\leq \frac{1}{6}-s\kappa R^{1/d}.
\ee
This allows us to rewrite
\be
\tilde p^{(m+1)}_{\rm in} \leq 3 \left(\tilde p^{(m)}_{\rm in}\right)^2.
\ee
The solution of the recurrence relation after $k$ cycles as
\be
\tilde p^{(k)}_{\rm out}\leq \tilde p^{(k+1)}_{\rm in}\leq\frac 1{3}\left(3 \tilde p^{(0)}_{\rm in}\right)^{2^k},
\ee
This result provides us with the bound on the required number of cycles as
\be
k \leq \log_2 \Biggl(\frac{\log 3\,(p_{\rm out}-s\kappa R^{1/d})}{\log 3\,(p_{\rm in}-s\kappa R^{1/d})}\Biggl)
\ee
Then, the in/out ratio satisfies
\be
R = 3^{k} \leq \left(\frac{\log 3\,(p_{\rm out}-s\kappa R^{1/d})}{\log 3\,(p_{\rm in}-s\kappa R^{1/d})}\right)^{\alpha} \leq \left(\frac{\log [3\,(p_{\rm out}-s\kappa R^{1/d})]}{\log [3 p_{\rm in}]}\right)^{\alpha}\quad \text{if} \quad p_{\rm in}\leq\frac 1{12}-\frac 12s\kappa R^{1/d},
\ee
where the exponent $\alpha = \log_23$.

Finally, similar to the previous Lemma, we merge the aforementioned two results to derive a hybrid bound on $R$. We construct our target compressor as a sequence of three smaller compressors. The first compressor goes from $z_{\rm in}\geq c(\kappa\eta)^{1/2}$ to a certain $\eta_1$ and has a ratio of $R_1$, the second compressor goes from a constant polarization $z_1\leq c^2 (1-g_0 a)/b$ to a constant $z_2> \frac{11}{12}+\frac 12s\kappa R^{1/d}$ and has a ratio of $R_2$, and the last compressor goes from a constant $z_2$ to a certain $z_{\rm out}\leq 1-s\kappa$ and has a ratio of $R_3$. For the ratio $R_1$ we use the bound Eq.~\eqref{eq:noisy_r_bound}, while for $R_3$ we use the Eq.~\eqref{eq:small_p_R}. The ratio $R_2$ is independent of both $z_{\rm in}$ and $z_{\rm out}$ and can be treated as a constant. The resulting total in/out ratio is
\be\label{eqs:noisy_input_output_r}
R = R_1R_2R_3\leq C \frac{\left(-\log 3[1-z_{\rm out}-O(g)]\right)^{\alpha}}{z_{\rm in}^{\mu_0+O(g)}}
\ee
for some $C>0$. This expression concludes our proof.
\end{proof}

\subsection{Cooling procedure and RESET process}
\label{sec:reset_process}

Once we have figured out the compression process, we need to find the cooling procedure so that the resulting depth and number of ancilla are compatible with the rest of the circuit. In particular, we want the implementation depth $T_{\rm RES}$ to satisfy $\kappa T_{\rm RES}\ll1$, otherwise the rest of the qubits will decohere before the RESET operation is complete. We also require that $\kappa N_a^{1/d}\ll 1$ to ensure that one can perform the circuit operations by naviagating the target qubits around the ancilla using SWAP operations.

The possibility of such a cooling schedule and the entire RESET process can be established using the following Lemma. Here, we employ standard big-Theta notation, $\Theta(x)$, along with big-O notation, $\tilde{O}(x)$, which omits logarithmic factors.

\begin{lem}\label{lem:cqc_circ}
Consider a $d$-dimensional noisy circuit composed of unitary operations followed by single-qubit noise described by a general contracting channel $\mathcal{N}$ as in Eq.~\eqref{eq:assump_contr} with a fixed-point purity $ \eta > 0$ and a strength $\|\mathcal{N} - \mathcal{I} \|_\diamond = \kappa > 0$ satisfying $ \kappa = O(\eta^{\mu / d})$, where $\mu = \log 3/\log(3/2) + \delta$ for any $\delta > 0$.  Such a circuit can implement a noisy RESET operation that outputs a single qubit with the following polarization
\be\label{eq:z_out_reset}
z_{\rm out} =  1- \Theta\Bigl((\kappa\eta^{-\mu/d})^{d/(\mu+d)}\Bigl)
\ee
using $N_a$ ancilla qubits and depth $T_{\rm RES}$ satisfying
\be
N_a = O\left(\frac{\log^{\alpha}(\kappa\eta^{-\mu/d}) }{(\kappa\eta)^{\mu d/(\mu+d)}}\right), \qquad \kappa T_{\rm RES} = \tilde O\Bigl((\kappa\eta^{-\mu/d})^{d/(\mu+d)}\Bigl),
\ee
for $\alpha = \log_23$.
\end{lem}

\begin{proof} We start by describing the settling process. Consider maximally-mixed ancilla states that are subjected to $T\gg 1$ layers of circuit evolution.  Using the expression we derived in Lemma~\ref{lem:settling} (see Eq.~\eqref{eq:settling_process}) 
from general contracting properties of the noise channel $\mathcal N$, we get
\be\label{eq:z_in_setimate2}
z_{\rm in} = \eta_{T} =  \Omega(\kappa\eta T).
\ee
It is convenient 
to choose the settling time as
\be\label{eq:cooling_time}
T \propto (\kappa\eta)^{-\mu/(\mu+d)},
\ee 
where the symbol $\propto$ denotes a relation up to an arbitrary multiplicative factor, and $\mu>0$. 
As we see below, this choice of settling time can be made equal to the time needed for swapping ancilla qubits, and thus minimizing the overall reset time.  Then the total depth of the cooling part can be bounded as
\be
\kappa T \propto (\kappa\eta^{-\mu/d})^{d/(\mu+d)}.
\ee

Therefore, the polarization from Eq.~\eqref{eq:z_in_setimate} takes the form
\be\label{eq:cooling_z_in}
z_{\rm in} = \Omega\Bigl((\kappa\eta)^{d/(\mu+d)}\Bigl).
\ee
Consider CQC that outputs a single qubit, $N_{\rm out} = 1$ using full stack of ancilla qubits as an input, $N_{\rm in} = N_{a}$  
Using Lemma~\ref{lem:noisy_lemma} with $z_{\rm in}$ in Eq.~\eqref{eq:cooling_z_in} and $z_{\rm out}$ in Eq.~\eqref{eq:z_out_reset}, the total number of ancilla required for the CQC is
\be
N_a \equiv R = O\left(\frac{\log^{\alpha}(\kappa\eta^{-\mu/d}) }{(\kappa\eta)^{(\mu_0+O(g))d/(\mu+d)}}\right).
\ee
To fix the value of $\mu$, we note that
\be
\kappa R^{1/d} = F\cdot (\kappa\eta)^{(\mu-\mu_0+O(g))/(\mu+d)}
\ee
where the prefactor $F$ takes the asymptotic form
\be  F = O\left((\kappa\eta^{-\mu/d})^{d/(\mu+d)}\log^\alpha(\kappa\eta^{-\mu/d})\right)
\ee
Then  there is a choice of $\delta>0$ such that using $\mu = \mu_0 + \delta$ makes the expression above to satisfy $\kappa R^{1/d}<g$. Moreover, taking arbitrarily small $g$, we can use arbitrarily small $\delta$. Thus
\be
N_a =O\left(\frac{\log^{\alpha}(\kappa\eta^{-\mu/d}) }{(\kappa\eta)^{\mu d/(\mu+d)}}\right).
\ee
The total depth of the compressor is
\be
T_{\rm RES} = T + O(R^{1/d}\log R),
\ee
where $T_{\rm RES}$ is the total depth of the RESET operation, and $T$ is the depth required for settling, as shown in Eq.~\eqref{eq:cooling_time}. Here, we use the fact that the natural logarithm $\log$ and $\log_3$ are related by a constant factor, which can be omitted in big-O notation. Then we arrive at the epxression
\be
\kappa T_{\rm RES} = \tilde O\Bigl((\kappa\eta^{-\mu/d})^{d/(\mu+d)}\Bigl), 
\ee
where $\tilde O$ ignores logarithmic corrections.
This expression concludes our proof.
\end{proof}

\begin{cor} \label{thm:reset}  
Consider an $N$-qubit, $d$-dimensional noisy unitary circuit of depth $D$, where each qubit is subject to single-qubit noise described by a channel $\mathcal{N}$ with a fixed-point purity $ \eta > 0$ and a strength $\|\mathcal{N} - \mathcal{I} \|_\diamond = \kappa > 0$ satisfying $ \kappa = O(\eta^{\mu / d})$, where $\mu = \log 3/\log(3/2) + \delta$ for any $\delta > 0$. Then, such a circuit can simulate an $ N' $-qubit, $ d $-dimensional noisy circuit of depth $ D' $, incorporating unitary gates and RESET operations subject to a single-qubit noise $\mathcal N'$ with a strength
\be\label{eq:effective_kappa}
\kappa' := \|\mathcal N'-\mathcal I\|_\diamond = \tilde O\left((\kappa\eta^{-\mu/d})^{d/(\mu+d)}\right),
\ee
with overhead
\be\label{eq:reset_overhead}
N/N' =   O\left((\kappa\eta)^{-\mu d/(\mu+d)}{\rm polylog}(\kappa\eta^{-\mu/d}) \right), \qquad D/D' =   O\left((\kappa\eta)^{-\mu/(\mu+d)}\right).
\ee
\end{cor}

This result directly follows from the fact that RESET produces qubits with polariztion in Eq.~\eqref{eq:z_out_reset} and the fact that the rest of the qubits accumulate the noise
\be
\kappa' = \kappa T_{\rm RES} + O(\kappa N_a^{1/d}) =  \tilde O\left((\kappa\eta^{-\mu/d})^{d/(\mu+d)}\right).
\ee
The overhead can be computed as $N'/N = N_a$ because each circuit qubit must have $N_a$ auxiliary qubits to implement RESET operations. The depth overhead can be computed as $D'/D \leq T_{\rm RES} O(\kappa N_a^{1/d})$ since, in the worst case, each layer may contain a RESET requiring depth $T_{\rm RES}$.

\subsection{Connection between RESETs and fault tolerance}

In this section, we prove the threshold theorem that quantifies the overhead in simulating a noiseless unitary circuit  of depth $D'$ up to some error $\varepsilon$ using a $D$-depth noisy circuit with RESETs when the noise strength is below some certain constant value.  

\begin{theoremS} [Threshold Theorem] \label{thm:threshold} 
There exists a constant $\kappa_0>0$ such that for all $\kappa: = \|\mathcal N-\mathcal I\|_{\diamond}<\kappa_0$ a $N$-qubit $d$-dimensional noisy unitary circuit with RESETs of depth $D$ can simulate $N'$-qubit $d$-dimensional noiseless unitary circuit of depth $D'$ up to error $\epsilon$ with overhead
\be\label{eq:thresholN_overhead}
\begin{split}
&D/D'\sim N/N' = O({\rm log}(N'D'/\epsilon)).
\end{split}
\ee
\end{theoremS}

Before proving this theorem, we begin by introducing some definitions. First, consider the standard notation of a $[[n,1,3]]$ quantum code, which encodes a single logical qubit into $n$ physical qubits and is capable of tolerating any single-qubit error \cite{gottesman1997stabilizer}. The we use this code to build a concatenated code. For each encoded subset, we assign $n_A$ ancilla qubits, which are necessary for performing error correction and implementing fault-tolerant gates. Consequently, the code and ancilla together form a block consisting of $n_B = n + n_A$ qubits. Thus, given a fixed number of qubits $N'$ in the target noiseless circuit, we consider the noisy circuit containing $N = N' n_B^L$ qubits, where $L \geq 0$ represents the total number of concatenation levels.

Next, we define the logical operators for each concatenation level. First, we set $P^{0}_0 := I$, $P^{0}_1 := X$, $P^{0}_2 := Y$, and $P^{0}_3 := Z$ as the single-qubit Pauli matrices. For each concatenation level $L \geq m \geq 0$, we then define the level-$m$ logical Pauli operators $P^{m}_{\pmb{\alpha}}$,
\be
P^m_{\pmb{\alpha}} : = \bigotimes_{i=1}^{N_m}P^{m}_{\alpha_i},
\ee
where $\pmb{\alpha} \in \mathbb{Z}_4^{N_m}$ is an integer-valued vector with components $\alpha_i \in \{0,1,2,3\}$, $N_m = N' n_B^{L-m}$ is the number of logical qubits at $m$th level, and we define level-$m$ Pauli operators for $m\geq 0$ through the recurrence relation
\be
P^{m+1}_{\alpha} := \bigl(P^{m}_{\alpha}\bigl)^{\otimes n} \otimes \bigl(P_0^{m} \bigl)^{\otimes n_A},
\ee
where we have chosen first $n$ qubits as encoding qubits and the rest $n_A$ qubits as ancilla.

After defining the logical operators, we define the logical superoperators that describe quantum channels. First, for any $k$-qubit channel $\mathcal{E}(\cdot)$ acting on a physical qubits subset $\Omega$ and any pair of vectors $\pmb{\alpha}, \pmb{\beta} \in \mathbb{Z}_4^k$, we define $T_{\pmb{\alpha} \pmb{\beta}}$ as the coefficients of the (Pauli) transfer matrix of the channel such that
\be
\forall \pmb{\alpha}\in \mathbb Z_4^{k}, O\in \mathsf O_{F\setminus\Omega}: \qquad \mathcal E^\dag(P^0_{\pmb{\alpha}}\otimes O) = \sum_{\pmb{\beta}\in \mathbb Z_4^{k}}T_{\pmb{\alpha}\pmb{\beta}}P^0_{\pmb{\beta}}\otimes O .
\ee
where $O$ belong to the set of operators $\mathsf O_{F\setminus\Omega}$ acting on the subset $F\setminus \Omega$ complementary to $\Omega$.
Transfer matrix uniquely defines the quantum channel acting on physical qubits. Then we define a quantum channel $\mathcal E^{(m)}(\cdot)$ to be a level-$m$ logical extension of $\mathcal E(\cdot)$ if its action follows $T_{\pmb{\alpha}\pmb{\beta}}$ such that
\be\label{eqs:m-elevl channel}
\forall \pmb{\alpha}\in \mathbb Z_4^{k}, O\in \mathsf O_{F\setminus\Omega}:\qquad \mathcal {E}^{(m)\dag} (P^{m}_{\pmb{\alpha}}\otimes O)= \sum_{\pmb{\beta}\in \mathbb Z_4^k}T_{\pmb{\alpha}\pmb{\beta}}P^{m}_{\pmb{\beta}}\otimes O.
\ee
There may be multiple extensions of a given Pauli channel, distinguished by their action on Pauli operators other than $ P^{m}_{\pmb{\alpha}}$. 

Finally, let us define the logical gates. Before that, to simplify the expressions, we define extend the traditional product notation to the ordered product of superoperators as follows.
\begin{notn} We denote the multiplication operator according to an ordered list of operations as
\be
\prod_{l\in (l_1,\dots, l_N)}{\mathcal O}_l := {\mathcal O}_{l_N} \dots {\mathcal O}_{l_1}.
\ee
\end{notn}
We also use a notation for an unspecified single-qubit noise model of given strength.
\begin{notn} \label{not:noise_channel} By $\mathcal N^{(m,\kappa)}$ we denote an unspecified CPTP map that satisfies two conditions:
\begin{enumerate}
    \item $\mathcal N^{(m,\kappa)} = \mathcal N^{(m)}_1\otimes \dots \otimes \mathcal N^{(m)}_{N_m}$, where $\mathcal N^{(m)}_i$ is a $m$-level extension of a single-qubit channel  $\mathcal N_i$.
    \item  $\|\mathcal N_i-\mathcal I\|_{\diamondsuit}\leq \kappa$;
\end{enumerate}
\end{notn}
Using these definitions, we now outline the properties of the gates used in the evolution. By the Solovay-Kitaev theorem, one can argue that any circuit can be represented with a limited set of quantum gates up to an arbitrarily small error. Therefore, without loss of generality, we assume that the target circuit is already composed of such gates and disregard compilation errors.

The logical extension of the standard gates must satisfy a notion of fault tolerance, which can be formulated as follows.
Note that the operations must also include the RESET.

\begin{dfn}[Fault-tolerant logical operation] \label{efn:transv_op} We call a subcircuit represented by a list of gate indices $O$ a fault-tolerant logical operation if the noisy and noiseless operations defined as
\be\label{eq:log_operations}
\mathcal O^{(m)} = \prod_{i\in O} {\mathcal O}^{(m-1)}_i, \qquad \tilde {\mathcal O}^{(m)} = \prod_{i\in O} {\mathcal N}^{(m-1,\kappa)} {\mathcal O}^{(m-1)}_i
\ee
for any $\kappa, \kappa'\geq0$ satisfy the condition
\be\label{eq:transvers_gate_def}
\tilde{\mathcal O}^{(m)} {\mathcal N}^{(m-1,\kappa')} = {\mathcal N}^{(m-1,\kappa'+\kappa |O|)} \mathcal  O^{(m)},
\ee
where $|O|$ is the length of the list $O$ (i.e. the number of operations in $O$).
\end{dfn}

In addition to logical operations, the logical circuit requires operations that mitigate the effect of errors. We formulate the main property of such a quantum error correction (QEC) gadget as follows.

\begin{dfn}[QEC gadget] \label{dfn:qec_gadget} We call a subcircuit represented by a list of gate indices $C$ an quantum error correction (QEC) gadget if for the operation defined as
\be
\tilde {\mathcal C}^{(m)} := \prod_{i\in C} {\mathcal N}^{(m-1,\kappa)} {\mathcal O}^{(m-1)}_i
\ee
there exists $c>0$ such that for all $\kappa,\kappa'\geq0$ the condition
\be
\tilde{\mathcal C}^{(m)} {\mathcal N}^{(m-1,\kappa')}\mathcal P^m =  {\mathcal N}^{(m-1,\kappa|C|)} {\mathcal N}^{(m,c(\kappa'^2+\kappa^2))}\mathcal P^m ,
\ee
 is staisfied, where $|C|$ is the length of the list $C$, $\mathcal P^m = \sum_{\nu=1}^{2^{N_m}} |\nu,m\>\<\nu,m|$ is the projector on the $m$-level logical space, and $|\nu,m\>$ are $m$-level logical codewords. 
\end{dfn}

The construction of logical gates and QEC gadgets using only unitary gates and RESET operations is well established in the literature. Here, we summarize the results with the following proposition.

\begin{prop}\label{prop:existense_steane_code}
For $[[7,1,3]]$ Steane code there exists a complete set of fault-tolerant operations, according to Definition~\ref{efn:transv_op}, that includes RESET and a computationally universal set of gates, as well as a QEC gadget according to Definition~\ref{dfn:qec_gadget}.
\end{prop}

\begin{proof}Consider a universal set of gates consisting of the $S$-gate, Hadamard gate, CNOT, $T$, and Toffoli gate. For the $[[7,1,3]]$ Steane code, the Hadamard, $S$-gate, and CNOT gates are transversal, i.e., they can be implemented as a product of single-qubit or pairwise two-qubit operations \cite{zeng2011transversality}. As a result, Eq.~\eqref{eq:transvers_gate_def} is satisfied automatically. A scheme for devising a fault-tolerant $T$-gate and Toffoli gate is provided in Ref.~\cite{boykin2010algorithms}. The measurement-free RESET and QEC gadget are presented in Ref.~\cite{heussen2024measurement}.
\end{proof}

Using the notation and definitions introduced above, we can proceed to the proof of Theorem~\ref{thm:threshold}, which constitutes the main result of this section. The proof employs similar ideas to those in previous works \cite{gottesman2000fault, aliferis2005quantum}, but explicitly avoids using mid-circuit measurements.

\begin{proof} Consider $[[7,1,3]]$ Steane code and a family noiseless concatenated logical circuit $\mathcal U$ consisting of two parts: (i) initialization of each $L$-level logical qubit $q$
in 
$|0\rangle$ state using $L$-level encoding operation $\mathcal E^L_q$ and (ii) a sequence of $L$-level fault-tolerant logical operations ${\mathcal O}^L_k$ including logical unitary gates and logical RESETs.
Part (ii) can be divided into $N_L\geq 0$ layers. The initialization and each operation in turn is represented as a combination of $(L-1)$-level logical operations forming $N_{L-1}$ layers, and so on. As a result, for any level of concatenation, there exists a circuit representation that implements the target unitary operation,
\be\label{eq:noiseless_ft_circ}
\begin{split}
\forall m\in[L], \quad  \exists\, D_m\geq 0, d_{mr}\geq 1,\, &M^m_r\in \mathbb Z_S^{d_{m r}}: \qquad \\
&\rho_{\rm out} :=\mathcal U(|0\>\<0|^{\otimes N}) \equiv \prod_{r=1}^{D_m} \left(\prod_{k\in M^m_r} {\mathcal O}^{(m)}_k \right) \prod_{q\in Q_m}  {\mathcal E}^{(m)}_q (|0\>\<0|^{\otimes N}),
\end{split}
\ee
where  $Q_m = [N'n_B^{L-m}]$ is the list of $m$-level logical qubits, ${\mathcal O}^{(m)}_k$ are logical gates being $m$-level extension of physical gates ${\mathcal O}^{(0)}_k$ that satisfy the recurrence 
relation in Eq.~\eqref{eq:log_operations}, ${\mathcal E}^{(m)}_q$ is $m$-level encoding operation for a logical zero state,
$M^m_r$ is the list of parallel $m$-level logical gates indices for layer $r$, $D_m$ is the number of logical gate layers, $d_{mr}$ is the number of gats in layer $r$, $S$ is the total set of fault-tolerant operations including RESET. We assume that $M^L_r$ exactly matches the decomposition of the target noiseless circuit we want to simulate with the noisy circuit.

Next, consider a family of noisy \textit{implementations} of the process in Eq.~\eqref{eq:noiseless_ft_circ}. For $m$-level implementation, we conside the output and the process
\be\label{eq:noisy_ft_circ}
\begin{split}
\tilde \rho^{(m)}_{\rm out} := \mathcal U_m (|0\>\<0|^{\otimes N}) \equiv \prod_{r=1}^{D_m} \left(\prod_{q\in Q_m} \tilde {\mathcal C}^{(m)}_q \prod_{k\in \tilde M^m_r} \tilde {\mathcal O}^{(m)}_k \right) \prod_{q\in Q_m}  \tilde {\mathcal E}^{(m)}_q (|0\>\<0|^{\otimes N}),
\end{split}
\ee
where $\tilde {\mathcal C}^{(m)}_q$ are noisy QEC gadgets (see Definition~\ref{dfn:qec_gadget}) and $\tilde {\mathcal O}^{(m)}_k$ are noisy $m$-level logical operations that satisfy the recurrence 
relation
\be\label{eq:O-noisy-def}
\tilde {\mathcal O}^{(m)}_j := \prod_{i\in O_j} {\mathcal N}^{(m-1,\kappa_m)} {\mathcal O}^{(m-1)}_i, \qquad \mathcal O^m_j := \prod_{i\in O_j} {\mathcal O}^{(m-1)}_i,
\ee
where ${\mathcal N}^{(m,\kappa_m)}$ denotes an unspecified noise channel (see Notation~\ref{not:noise_channel}) of $m$-dependent noise strength $\kappa_m\geq0$. 
Our goal is to find proper strength values $\kappa_m$ and sets $\tilde M_r^m$ such that we can relate $\rho^{(m)}_{\rm out}$ to $\rho^{(0)}_{\rm out}$ in Eq.~\eqref{eq:noisy_ft_circ} and to the noiseless circuit output $\rho_{\rm out}$ in Eq.~\eqref{eq:noiseless_ft_circ}. 

Let us start our analysis at concatenation level $m=1$, where we set $\kappa_1 \equiv \kappa$ to be the physical noise strength. The product of noisy fault-tolerant operations in can be represented as a combination of noisy physical operations and satisfies
\be\label{eq:ps9edu}
\tilde {\mathcal O}^1_j\mathcal P^1 = \left(\prod_{i\in O_j} {\mathcal N}^{(0,\kappa)} {\mathcal O}^{(0)}_i\right)\mathcal P^1 = {\mathcal N}^{(0,N_O\kappa)}\left(\prod_{i\in O_j} {\mathcal O}^{(0)}_i\right)\mathcal P^1  = {\mathcal N}^{(0,N_O\kappa)} {\mathcal O}^1_j\mathcal P^1 = {\mathcal N}^{(0,N_O\kappa)}\mathcal P^1 {\mathcal O}^1_j,
\ee
where ${\mathcal O}_j^1 \equiv \prod_{i\in O_j} {\mathcal O}^{(0)}_i$ is  noiseless 1-level logical operation, $O_j$ is the list of all gates in the logical operation ${\mathcal O}_j^1$, and $N_O := \sup_j |O_j|$ is the largest number of gates in any such logical operation. Here, in the second transformation, we used the property of operation fault-tolerance (see Definition~\ref{efn:transv_op}). Similar statement can be done about the noisy RESET operations
\be
\begin{split}
\prod_{q\in Q_1}\mathcal {\tilde E}^{(1)}_q (|0\>\<0|^{\otimes N}) &:= \prod_{q\in Q_1}\prod_{i\in E_q} \mathcal N^{(0,\kappa)} \mathcal O^{(0)}_i(|0\>\<0|^{\otimes N}) \\
&= \mathcal N^{(0,N_O\kappa)}\prod_{q\in Q_1}\prod_{i\in E_q} \mathcal O^{(0)}_i (|0\>\<0|^{\otimes N}) = \mathcal N^{(0,N_O\kappa)} \mathcal P^1 \prod_{q\in Q_1}\mathcal E^{(1)}_q (|0\>\<0|^{\otimes N})
\end{split}
\ee
where $\mathcal E^{(1)}_q$ is the noiseless 1-level encoding operation and $E_q$ is the list of all gates in $\mathcal E^{(1)}_q$. The last transition is justified by the fact that noiseless RESET operation $\mathcal E^{(1)}_q$ a $m$-level pure logical state $|0\>\<0|_q$ satisfying $|0\>\<0|_q = \mathcal P^1(|0\>\<0|_q)$. Next, the corrected $m$-level logical operations satisfy
\be
\begin{split}
\prod_{q\in Q_1} \mathcal{\tilde C}^1_q \prod_{j\in \tilde M^1_r} \mathcal {\tilde O}^1_i\, \mathcal N^{(0,N_O\kappa)}\mathcal P^1 & = \prod_{q\in Q_1} \mathcal{\tilde C}^1_q\, \mathcal N^{(0,2N_O\kappa)}\mathcal P^1 \prod_{k\in \tilde M^1_r} \mathcal {O}^{(1)}_i \\
&=  \mathcal N^{(0,N_O\kappa)} \mathcal P^1\mathcal N^{(1,c(4N^2_O+1)\kappa^2)}
 \prod_{k\in \tilde M^1_r} \mathcal O^{(1)}_k
\end{split}
\ee
This expression leads to the expression for the noisy circuit channel
\be\label{eq:soidcp78}
\tilde \rho^2_{\rm out} = \mathcal N^{(0,N_O\kappa)} \prod_{r=1}^{N_1} \left( \prod_{k\in \tilde M^1_r} \mathcal N^{(1,c(4N^2_O+1)\kappa^2)}\mathcal O^{(1)}_k \right)\prod_{q\in Q_1}\mathcal E_q^{(1)}(|0\>\<0|^{\otimes N})
\ee
For any desirable set of layer instructions $\{\tilde M_r^2\}$, there exist number of parallel layers $D_2<D_1$ and a set of instructions $\{\tilde M^1_r\}$ such that
 \be\label{eqs:0ipdojc}
\prod_{r=1}^{D_1} \left( \prod_{k\in \tilde M^1_r} \mathcal O^{(1)}_k \right) \prod_{q\in Q_1} \mathcal E_q^{(1)} = \prod_{r=1}^{D_2} \left( \prod_{q\in Q_2} {\mathcal C}^{(2)}_q\prod_{k\in \tilde M^2_r} \mathcal O^{(2)}_k \right) \prod_{q\in Q_2} \mathcal E^{(2)}_q 
 \ee
To see this, we can just also decompose $\mathcal E^{(2)}_q = \left(\prod_{i\in E_q} \mathcal O^{(1)}_i\right)\prod_{q\in Q_1} \mathcal E_q^{(1)}$ and $\mathcal O^{(2)}_j=\sum_{i\in O_j} \mathcal O^{(1)}_i$ and combine all gates in the new set of instructions. Using equivalence in Eq.~\eqref{eqs:0ipdojc}, the expression in Eq.~\eqref{eq:soidcp78} becomes
 \be\label{eqs:s09dc9ikmm,m}
\begin{split}
\tilde {\mathcal U}_1(|0\>\<0|^{\otimes N})=\mathcal N^{(0,N_O\kappa)} \prod_{r=1}^{N_2} \left(\prod_{q\in Q_2} \tilde {\mathcal C}^{(2)}_q \prod_{k\in \tilde M^2_r} \tilde {\mathcal O}^{(2)}_k \right) \prod_{q\in Q_2}  \tilde {\mathcal R}^{(2)}_q(|0\>\<0|^{\otimes N})= \mathcal N^{(0,N_O\kappa)} \tilde{\mathcal U}_2(|0\>\<0|^{\otimes N})
\end{split}
\ee
where $\tilde {\mathcal O}^2_k$ satisfies Eq.~\eqref{eq:O-noisy-def} with $\kappa_2 = c(4N^2_O+1)\kappa^2$. Repeating steps in Eqs.~\eqref{eq:ps9edu}-\eqref{eqs:s09dc9ikmm,m} replacing $1\to m$ and $2\to m+1$, we get the recurrence relation
\be
\rho^{(m)}_{\rm out}:=\tilde {\mathcal U}_{m}(|0\>\<0|^{\otimes N}) = \mathcal N^{(m-1,\kappa_{m})}\tilde {\mathcal U}_{m+1}(|0\>\<0|^{\otimes N}) = \mathcal N^{(m-1,\kappa_{m})}\rho^{(m+1)}_{\rm out}.
\ee
for $\kappa_m := N_O\mu^{-1} (\mu\kappa)^{2^{m-1}}$ for $m>1$ and $\mu = c(4N_O^2 + 1)$. This recurrence relation has a solution
\be
\tilde {\mathcal U}_{1}(|0\>\<0|^{\otimes N}) = \left(\prod_{m\in[L]} \mathcal N^{(m-1,\kappa_m)}\right)\tilde {\mathcal U}_{L-1}(|0\>\<0|^{\otimes N})
\ee
Next, we choose the set $\{\tilde M^{L-1}_r\}$ such that
 \be
\prod_{r=1}^{N_{L-1}} \left( \prod_{k\in \tilde M^{L-1}_r} \mathcal O^{(L-1)}_k \right) \prod_{q\in Q_L} \mathcal E^{(1)}_q  = \prod_{r=1}^{D} \left(\prod_{k\in M^L_r} \mathcal O^{(L)}_k \right) \prod_{q\in Q_L} \mathcal E^L_q 
 \ee
Using the fault-tolerance of the operations, we get
\be
\tilde{\mathcal U}_L = \mathcal N^{(L,G\kappa_{L})}{\mathcal U},
\ee
where $G = N'(D'+1)$ is the total number of gates in the operations in the $L$-level representation of the circuit $\mathcal U$, including RESETs and excluding QEC gadgets (since they do not affect $L$-level logic errors).
\be
\tilde \rho^1_{\rm out} = \left(\prod_{m\in[L-1]} \mathcal N^{(m,\kappa_m)}\right)\mathcal N^{(L,G\kappa_{L})}\rho_{\rm out},
\ee
where $G = N'D'$ is the total number of gates in the operations in $L$-level representation of the circuit. This expression connects the output physically implementable circuit $\tilde \rho^1_{\rm out}$ in Eq.~\eqref{eq:noisy_ft_circ} and the output of the logical circuit $\rho_{\rm out}$ in Eq.~\eqref{eq:noiseless_ft_circ}. Without the loss of generality, we consider individual qubit measurements in z-basis, i.e.
\be
z_q = \Tr (Z_q^{(L)}\tilde \rho_{\rm out}^1)
\ee
wjere $Z^{(L)}$ is $L$-level Pauli-Z operator applied to qubit $q$.

 For concatenated Steane code, the errors for levels $m\leq L-1$ can be mitigated by taking a majority vote over logical block for each logical qubit. Thus, the probability that the channel $\mathcal N^{(m,\kappa_m)}$ will affect the value of a logical level-$L$ qubit measurement is given by 
\be
\epsilon_m = O\left(\kappa_m^{\frac 12 (n^{L-m}+1)}\right) = O\left(\bigl[N_O\mu^{-1} (\mu\kappa)^{2^{m-1}}\bigl]^{\frac 12 (n^{L-m}+1)}\right) = O\left((\mu\kappa)^{2^{L-2}}\right)
\ee
In the case $\kappa<\kappa_0:=\mu^{-1}$, all these errors are marginal compared to the logical error from $L$th level of concatenation that can be expressed as
\be
\epsilon = O(N'D'\kappa_L) = O(N'D'(\mu\kappa)^{2^L}).
\ee
Thus, the number of concatenation levels to achieve the target error $\epsilon$ is
\be
L = O(\log\log N'D'/\epsilon).
\ee
Taking into account that each block requires $n_B$ qubuts and each gadget in Proposition~\ref{prop:existense_steane_code} requires no more than $Q$ layers of gates, we get
\be
N/N' = n_B^L = O(\log(D'N'/\epsilon)),\qquad D/D' = Q^L = O(\log(D'N'/\epsilon)).
\ee
This expression concludes our proof.
\end{proof}

\subsection{Main result}
\label{sec:main_result}

In this section we prove our main results by combining results from previous sections. 

\begin{thrm}[Formal]\label{thm:main}
Consider quantum circuits with noise generated by a general contracting channel, $\mathcal N$, characterized by $\kappa =  \|\mathcal{N} - \mathcal{I}\|_{\diamond}$ and a fixed point $\sigma^*$ with purity parameter $\eta = \sqrt{2\Tr[(\sigma^*)^2] - 1}$. Then, for any $\eta>0$, $\epsilon>0$, and integer $d\geq 1$, there exists $\kappa = O(\eta^{\mu/d})$ for any $\mu > \mu_0 = \log 3/\log(3/2)$, such that an $N$-qubit, $d$-dimensional noisy circuit of depth $D$ can simulate the output of any $N'$-qubit, $d$-dimensional noiseless circuit of depth $D'$ with an error of at most $\epsilon$ and overhead
\be
\begin{split}
& N/N' =   O\left((\kappa\eta)^{-\mu d/(\mu+d)}{\rm polylog}(\kappa\eta^{-\mu/d}) {\rm log}(N'D'/\epsilon)\right),\\
& D/D' =   O\left((\kappa\eta)^{-\mu/(\mu+d)}{\rm log}(N'D'/\epsilon)\right).
\end{split}
\ee
\end{thrm}
\begin{proof}The proof of this theorem follows directly by combining Corollary~\ref{thm:reset} and Theorem~\ref{thm:threshold}.  Corollary~\ref{thm:reset} implies that there exists a local noisy  unitary circuit with noise strength $\kappa>0$ that can emulate any given local RESET-enabled noisy unitary circuit with a noise strength $\kappa' > 0$ in Eq.~\eqref{eq:effective_kappa}. Then, by invoking Theorem~\ref{thm:threshold}, we demonstrate that this RESET-enabled circuit can be employed to simulate the target unitary circuit up to error $\epsilon$ with an overhead of $O\left( \log \left( D' N'/\epsilon \right) \right)$, provided that $\kappa'$ is below a fixed constant threshold. According to the expression in Eq.~\eqref{eq:effective_kappa}, this condition can be satisfied for sufficiently small values of $\kappa = O(\eta^{\mu/d})$. Then, the total overheads in terms of the qubit count $N/N'$ and circuit depth $D/D'$ are obtained by taking the product of the respective overheads for error correction, given in Eq.~\eqref{eq:thresholN_overhead}, and for the implementation of RESET, as defined in Eq.~\eqref{eq:reset_overhead}. This yields the statement of this theorem. 
\end{proof} 

\section{Average-case complexity}

In this section, we prove our results corresponding to random quantum circuits in the presence of noise. 

\subsection{Algebra of subset distances}

We represent a noisy Haar random circuit as a sequence of gate operations, where each gate is described by a superoperator denoted as $\mathcal{C}_k(\cdot)$, which acts on a specific subset of qudits $\Omega_k$ and can be expressed in the following form:
\be\label{eq:noise_circle}
\mathcal C_k = \mathcal N_k\circ\mathcal U_k, \qquad \mathcal U_k(\rho) := U_k \rho U_k^\dag, \qquad \mathcal N_k(\rho) := \sum_{i=1}^{q^{2|\Omega_k|}-1} K_k^i \rho K_k^{i\dag},
\ee
where $U_k \sim \mathcal{B}$ is a unitary transformation acting on qudits in $\Omega_k$, sampled from the Haar distribution $\mathcal{B}$, $K_k^i$ are Kraus operators satisfying $\sum_i K_k^{i\dagger} K_k^i = I$, and $|\Omega_k|$ represents the number of qudits in the set $\Omega_k$. Following standard approximations, we assume that operators $K_k^i$ are independent of the gates $U_k$.

Then, the state of the qubits after applying $k$ gates takes the form
\be\label{eq:dth_state}
\rho_k := \mathcal{C}_k \circ \cdots \circ \mathcal{C}_1(\rho_{k=0}),
\ee
where $\rho_{k=0}$ denotes the input state.

Given an arbitrary pair of states $\rho_k$ and $\sigma_k$ as defined in Eq.~\eqref{eq:dth_state}, we focus on their \textit{subset distances}. First, we define the marginal trace distance for a subset $G$ as
\be\label{eq:trace_distance}
T_{G,k} := \frac{1}{2} \mathbb{E}_{\mathcal{B}} \bigl\| \Tr_{F \setminus G}(\rho_k - \sigma_k) \bigr\|_1 \equiv \frac{1}{2} \mathbb{E}_{\mathcal{B}} \Tr_{G} \bigl| \Tr_{F \setminus G}(\rho_k - \sigma_k) \bigr|,
\ee
where $F$ denotes the set of all qubits, $\Tr_A$ represents the partial trace over the subset of qubits $A$, $\mathbb{E}_{\mathcal{B}}$ denotes the gate average over the Haar distribution, $|O| := \sqrt{O^\dagger O}$ is the absolute value of the operator, and $\|M\|_1 := \Tr \sqrt{M^\dagger M}$ is the trace norm. The set-theoretic notation is illustrated in Fig.~\ref{fig:topology}.

\begin{figure}[t!]
    \centering
    \includegraphics[width=1\textwidth]{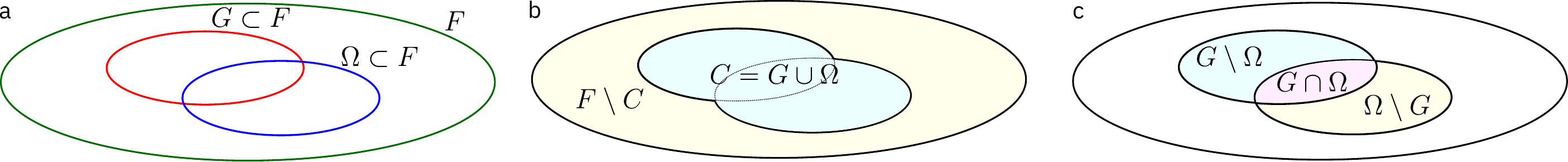}
    \caption{\textbf{Set theory notations.} Venn diagrams illustrating different subsets and their intersections. Panel a shows the complete set $F$ and the subsets $G \subset F$ and $\Omega \subset F$. Panel b shows the union of two sets $C = G \cup \Omega$ (blue) and the difference of two sets $F \setminus C$ (yellow). Panel c, along with the set differences $G \setminus \Omega$ (blue) and $\Omega \setminus G$ (yellow), also illustrates the intersection $G \cap \Omega$ of the subsets $G$ and $\Omega$ (red).}
    \label{fig:topology}
\end{figure}

We also define the marginal Hilbert-Schmidt (HS) distance as
\be\label{eq:HS_distance}
S_{G,k} := \frac{1}{2} \mathbb{E}_{\mathcal{B}} \bigl\| \Tr_{F \setminus G}(\rho_k - \sigma_k) \bigr\|^2_2 \equiv \frac{1}{2} \mathbb{E}_{\mathcal{B}} \Tr_{G} \bigl[ \Tr_{F \setminus G}(\rho_k - \sigma_k)^2 \bigr],
\ee
where $\|M\|_2 = \sqrt{\Tr M^\dagger M}$ is the Frobenius norm. The trace distance and the HS distance are related by the inequalities \cite{golub2013matrix}
\be\label{TS_ineq}
S_{G,k} \leq T_{G,k}, \qquad T^2_{G,k} \leq N_G S_{G,k},
\ee
where $N_G := q^{|G|}$ is the Hilbert space dimension of the subset $G$. The first inequality follows from the monotonicity of the Schatten norms \cite{raissouli2010various}, while the second inequality results from the Cauchy-Schwarz inequality. By default, $S_{\emptyset,k} = T_{\emptyset,k} = 0$, where $\emptyset$ is the empty set of qubits.

In the derivations, we use the HS distance as a convenient measure that can be averaged over the Haar distribution using the well-known expression for the fourth moments of random matrix elements. One can then identify simple linear transformations that describe the evolution of these distances. We summarize these transformations in the following Proposition.

\begin{prop}[Evolution of Hilbert-Schmidt distances] Suppose the $d$-th gate is a noisy Haar random gate in the form of Eq.~\eqref{eq:noise_circle} and applied to a subset $\Omega_k \equiv \Omega$. For any subset $G\subset F$, the HS distance in Eq.~\eqref{eq:HS_distance} satisfies
\be\label{lem_step_statement}
S_{G,k} = c_1\,S_{G\setminus\Omega,k-1}+c_2\,S_{G\cup\Omega,k-1},
\ee
where the coefficients are
\be\label{eq:c1c2def}
\begin{split}
c_1 = \frac{N_{G\cap\Omega}(N^2_{\Omega\setminus G}u-v)}{N_\Omega^2-1},\quad 
c_2 = \frac{N_{\Omega\setminus G}(N^2_{G\cap\Omega}v-u)}{N_\Omega^2-1},
\end{split}
\ee
 and noise parameters are
\be\label{eq:uv_params}
\begin{split}
&u = \frac{1}{N_{\Omega}N_{\Omega\setminus G}}\Tr_{G\cap\Omega}  \left|\Tr_{\Omega\setminus G} \sum_iK_k^i K_k^{i\dag}\right|^2,\\
&v = \frac{1}{N_\Omega N_{ G\cap\Omega}}\sum_{ij}\Tr_{\Omega\setminus G}  \left|\Tr_{\Omega\cap G} K_k^i K_k^{j\dag}\right|^2.
\end{split}
\ee
\end{prop}

\begin{proof} Since the entire proof concerns the effect of a single gate $k$, we can simplify the notation by omitting the index $k$, i.e., setting $K_{k}^i \equiv K^i$, $U_{k} \equiv U$, and $\Omega_{k} \equiv \Omega$. Combining the definition of the HS distance in Eq.~\eqref{eq:HS_distance}, the explicit form of the gate channel in Eq.~\eqref{eq:noise_circle}, and the set-theoretic notation (see Fig.~\ref{fig:topology}), we write
\be\label{eq:explicitSGd}
S_{G,k} = \mathbb{E}_U \Tr_G \Bigl( \Tr_{F \setminus G} K^i U \delta \rho_{k-1} U^\dagger K^{i \dagger} \Bigr)^2,
\ee
where $\delta \rho_k := \rho_k - \sigma_k$, $\mathbb{E}_U$ denotes the average over the gate unitary $U \sim \mathcal{B}_\Omega$, sampled from the Haar distribution over the subset $\Omega$. The sum is taken over $i = 1, \dots, 4^{|\Omega|} - 1$, as in Eq.~\eqref{eq:noise_circle}. Here and for the rest of the proof, we also adopt the ``silent sum'' convention: if an index appears twice in a single term, it implies summation over that index without explicitly writing the summation symbol.

Next let us divide the entire set on four complementary parts, $F\setminus (G\cup \Omega)$, $G\setminus \Omega$, $G\cap\Omega$, and $\Omega\setminus G$ (see Fig.~\ref{fig:topology} for illustration). Using this division, we can introduce two basis states defined on full qubit set $F$ and the subset $G$ respectively,
\be
\begin{split}
&|\xi A\alpha a\> := |\xi\>_{F\setminus (G\cup \Omega)}\otimes |A\>_{G\setminus \Omega}\otimes |\alpha\>_{G\cap\Omega}\otimes|a\>_{\Omega\setminus G},\\
&|\alpha a\> := |\alpha\>_{G\cap\Omega}\otimes|a\>_{\Omega\setminus G}.
\end{split}
\ee
In these expressions, $|\xi\rangle_{F\setminus (G\cup \Omega)}$, $|A\rangle_{G\setminus \Omega}$, $|\alpha\rangle_{G\cap\Omega}$, and $|a\rangle_{\Omega\setminus G}$ represent arbitrary basis states of qubits within the sets $F\setminus (G\cup \Omega)$, $G\setminus \Omega$, $G\cap\Omega$, and $\Omega\setminus G$. Using these basis states, we define the following three matrix elements
\be
\begin{split}
&[\delta\rho]_{\xi A\alpha a,\xi' B\beta b} := \<\xi A\alpha a|\delta\rho_{d-1}|\xi'B\beta b\>,\\
&K^i_{\alpha a,\beta b} := \<\alpha a|K^i|\beta b\>,\\
&U_{\alpha a,\beta b} := \<\alpha a|U|\beta b\>.\\
\end{split}
\ee
These notations allow us explicitly rewrite
Eq.~\eqref{eq:explicitSGd} as
\be\label{eq:SGN_using_matrix_elements}
\begin{split}
S_{G,k} = \mathbb E_U[U_{\beta'b',\beta b}&U^*_{\alpha' a',\alpha a}U_{\mu'm',\mu m}U^*_{\nu'n',\nu n}]\\
&\times K^i_{\gamma' g,\beta' b'}K^{i*}_{\gamma g,\alpha'a'}K^j_{\gamma g',\mu'm'}K^{j*}_{\gamma'g',\nu'n'}[\delta\rho]_{\xi B \beta b,\xi A \alpha a}[\delta\rho]_{\xi'A\mu m,\xi' B \nu n},
\end{split}
\ee
where used the ``silent sum'' convention.

Now we can use the standard expression for the fourth moment of matrix elements in a Haar ensemble
\be\label{eq:4thmoment}
\begin{split}
\mathbb E_U[U_{\beta'b',\beta b}U^*_{\alpha' a',\alpha a}U_{\mu'm',\mu m}U^*_{\nu'n',\nu n}] = \frac{1}{N_\Omega^2-1}\Biggl(&\delta_{\alpha\beta}\delta_{ab}\delta_{\mu\nu}\delta_{mn}\delta_{\alpha'\beta'}\delta_{a'b'}\delta_{\mu'\nu'}\delta_{m'n'}\\
&+\delta_{\alpha\mu}\delta_{am}\delta_{\beta\nu}\delta_{bn}\delta_{\beta'\nu'}\delta_{b'n'}\delta_{\alpha'\mu'}\delta_{a'm'}\\
&-\frac{1}{N_\Omega}\delta_{\alpha\beta}\delta_{ab}\delta_{\mu\nu}\delta_{mn}\delta_{\beta'\nu'}\delta_{b'n'}\delta_{\alpha'\mu'}\delta_{a'm'}\\
&-\frac{1}{N_\Omega}\delta_{\alpha\mu}\delta_{am}\delta_{\beta\nu}\delta_{bn}\delta_{\alpha'\beta'}\delta_{a'b'}\delta_{\mu'\nu'}\delta_{m'n'}
\Bigl),
\end{split}
\ee
where $\delta_{ab}$ is Kronecker delta and $N_\Omega := 2^{|\Omega|}$ is the Hilbert space dimension of the subset $\Omega$. Note that since each symbol is mentioned once in each term, the silent sum does not apply to this expression. Inserting Eq.~\eqref{eq:4thmoment} into Eq.~\eqref{eq:SGN_using_matrix_elements}, we get
\be
\begin{split}
S_{G,k} = \frac{1}{N_\Omega^2-1}\Biggl(&K^i_{\gamma' g,\alpha' a' }K^{i*}_{\gamma g,\alpha'a'}K^j_{\gamma g',\mu'm'}K^{j*}_{\gamma'g',\mu'm'}[\delta\rho]_{\xi B \alpha a,\xi A \alpha a}[\delta\rho]_{\xi'A\mu m,\xi' B \mu m}\\
& + K^i_{\gamma' g, \beta' b}K^{j*}_{\gamma' g',\beta'b'}K^j_{\gamma g',\alpha'a'}K^{i*}_{\gamma g,\alpha'a'}[\delta\rho]_{\xi B \beta b,\xi A \alpha a}[\delta\rho]_{\xi'A\alpha a,\xi' B \beta b}\\
&-\frac{1}{N_\Omega} K^i_{\gamma' g, \beta' b}K^{j*}_{\gamma' g',\beta'b'}K^j_{\gamma g',\alpha'a'}K^{i*}_{\gamma g,\alpha'a'}[\delta\rho]_{\xi B \alpha a,\xi A \alpha a}[\delta\rho]_{\xi'A\mu m,\xi' B \mu m}\\
&-\frac{1}{N_\Omega}K^i_{\gamma' g,\alpha' a' }K^{i*}_{\gamma g,\alpha'a'}K^j_{\gamma g',\mu'm'}K^{j*}_{\gamma'g',\mu'm'}[\delta\rho]_{\xi B \beta b,\xi A \alpha a}[\delta\rho]_{\xi'A\alpha a,\xi' B \beta b}\Biggl).
\end{split}
\ee

These products of matrix elements can be represented as follows,
\be
\begin{split}
S_{G,k+1} = \frac{1}{N_\Omega^2-1}\Biggl(&\Tr_{G\cap\Omega}  \Bigl|\sum_\alpha\Tr_{\Omega\setminus G} K^i K^{i\dag}\Bigl|^2 S_{G\setminus\Omega,k}+\sum_{ij}\Tr_{\Omega\setminus G}  \left|\Tr_{\Omega\cap G} K^i K^{j\dag}\right|^2S_{G\cup\Omega,k}\\
&-\frac{1}{N_\Omega}\sum_{ij}\Tr_{\Omega\setminus G}  \left|\Tr_{\Omega\cap G} K^i K^{j\dag}\right|^2 S_{G\setminus\Omega,k}-\frac{1}{N_\Omega}\Tr_{G\cap\Omega}  \Bigl|\sum_\alpha\Tr_{\Omega\setminus G} K^i K^{i\dag}\Bigl|^2S_{G\cup\Omega,k}\Biggl).
\end{split}
\ee

Next, we use the definitions of the parameters $u$ and $v$ given in Eq.~\eqref{eq:uv_params} to obtain
\be\label{eq:SGdalmostthere}
\begin{split}
S_{G,k+1} = \frac{1}{N_\Omega^2-1}\Biggl(\Bigl[N_{\Omega}N_{\Omega\setminus G}u-N_{ G\cap\Omega}v\Bigl]S_{G\setminus\Omega,k}+\Bigl[N_{\Omega}N_{ G\cap\Omega}v-N_{\Omega\setminus G}u\Bigl]S_{G\cup\Omega,k}\Biggl).
\end{split}
\ee

Finally, we consider that $\Omega \equiv (\Omega\setminus G)\cup (G\cap\Omega)$. Therefore, the dimension of the Hilbert space for the set $\Omega$ is given by
\be
N_\Omega = N_{\Omega\setminus G}N_{ G\cap\Omega}.
\ee
Inserting this expression into Eq.~\eqref{eq:SGdalmostthere}, we obtain the result in Eqs.~\eqref{lem_step_statement}-\eqref{eq:c1c2def}. This last step completes our proof.
\end{proof}

\subsection{Two-qudit gates case, uniform symmetric noise}

Consider the case of a circuit that consists of two-qudit gates and noise model that does not depend on the gate (i.e., uniform) and symmetric in respect of swapping the qubits.
We introduce the following two sets of parameters. The first set is
\be\label{eq:uv_2q_gate}
a = \frac{1}{q^4} \sum_{ij} \bigl| \Tr_\Omega K^i K^{j \dagger} \bigr|^2, \qquad b = \frac{1}{q^2} \Tr_\Omega \Bigl| \sum_i K^i K^{i \dagger} \Bigr|^2.
\ee
Let $\Omega_k = \{m,m'\}$, where $m$ and $m'\neq m$ are indices of two qubits that constitute the set $\Omega_k$. Then we define
\be\label{eq:uv_2q_gate2}
A = \frac{1}{q^2} \sum_{ij} \Tr_m\bigl| \Tr_{m'} K^i K^{j \dagger} \bigr|^2, \qquad B = \frac{1}{q^2} \Tr_m \Bigl| \Tr_{m'}\sum_i K^i K^{i \dagger} \Bigr|^2.
\ee
where $\Tr_{m}$ and $\Tr_{/m'}$ indicate partial traces taken over qudit $m$ and $m'$ of the set $\{m,m'\}$, respectively. Note that $A$ and $B$ are independent on the order of the qubits. Using these parameters, we analyze in detail the rules in Eq.~\eqref{lem_step_statement} for three different cases of alignment between the two-qudit gate subspace $\Omega$ and the subset $G$.

\subsubsection{Case 1: the gate applies within the subset}

Consider $\Omega \in G$, where the gate is applied within the qudit subset $G$. Then Eq.~\eqref{lem_step_statement} leads to
\be\label{eqs:same_set_rule}
S_{G,k} = (\alpha - r) S_{G \setminus \Omega, k-1} + r S_{G, k-1},
\ee
where we define the parameters
\be \label{eq:rule_1}
\alpha := \frac{q^2 a + b}{q^2 + 1}, \qquad r := \frac{q^4 a - 1}{q^4 - 1}.
\ee
Note that for a unitary noise channel we have $a = b = 1$, and therefore, $S_{G,k} = S_{G, k-1}$.

\subsubsection{Case 2: the gate intersects with the subset}

Consider $\Omega_k \cap G \neq \{\emptyset\}$ and $\Omega_k \setminus G \neq \{\emptyset\}$, and let there be two qudits $m$ and $m'$ such that $G \cap \Omega_k = \{m\}$ and $G \setminus \Omega_k = \{m'\}$. Then Eq.~\eqref{lem_step_statement} leads to
\be\label{eq:cross_gate_expr}
S_{G,k} = \frac{1}{2} (\beta + \mu) S_{G \setminus \Omega_k, k-1} + \frac{1}{2} (\beta - \mu) S_{G \cup \Omega_k, k-1},
\ee
where we defined
\be\label{eq:rule_2}
\beta := \frac{q}{q^2 + 1} (A + B), \qquad \mu := \frac{q}{q^2 - 1} (B - A).
\ee
Again, for a unitary noise channel, $A_{mm'} = B_{m'm} = 1$, we get
\be
S_{G,k} = \frac{q^2}{q^2 + 1} \left( S_{G \setminus \Omega_k, k-1} + S_{G \cup \Omega_k, k-1} \right).
\ee
This expression is exactly the same as for a noiseless circuit because multiplying a Haar-random unitary by any unitary does not change the distribution.

\subsubsection{Case 3: the gate applies outside the subset}

Finally, consider the case $\Omega \cap G =\{\emptyset\}$, i.e. the gate applies to the complementary subset $F\setminus G$. Then Eq.~\eqref{lem_step_statement} leads to a simple relation
\be
S_{G,k} = S_{G,k-1}.
\ee

\subsection{Logarithmic-depth lower bound for local circuits}
\label{sec:log_depth_bound}

Consider the case of qubits $q=2$. We recall that the \textit{parallel layer} is a part of the circuit where each qubit participates in one and only one gate.

\begin{thrm}[Restatement]
Consider a noisy random unitary circuit consisting of $D$ parallel layers of two-qubit gates. For any bit strings $\vec{z}, \vec{z}' \in \{0,1\}^n$, with $\vec{z} \neq \vec{z}'$, the corresponding initial quantum states $\rho = |\vec{z}\rangle \langle \vec{z}|$ and $\sigma = |\vec{z}'\rangle \langle \vec{z}'|$ satisfy
\be
\frac{1}{2} \mathbb{E}_{\mathcal B} \|\mathcal C(\rho - \sigma)\|_1 \geq e^{-\Gamma D},
\ee
where $\Gamma>0$ is a constant, $\mathcal{C}$ is a noisy random circuit channel, and $\mathbb{E}_{\mathcal B}$  denotes the expectation over the random two-qubit unitaries.
\end{thrm}

\begin{proof} Consider the marginal HS distance $S_{\{i,j\},l}$ between states $\rho$ and $\sigma$ after $0 \leq l \leq D - 1$ parallel layers of the circuit for any pair of qubits $i$ and $j$. Let us define the quantity
\be
S^{(2)}_l :=  \max_{\{i,j\}\in C_{l+1}} S_{\{i,j\},l},
\ee
where $C_l$  (layer connectivity) is the set of all two-qubit pairs on which gates at the layer $l$ act on.
Let us now consider the evolution of $S_{\{i,j\},k}$. For parallel circuit, depending on the choice of qubit group and the layer, there are two possible scenarios.

In the first scenario, qubits $i$ and $j$ are subject to the same two-qubit gate. Then, we can use Eq.~\eqref{eqs:same_set_rule} and rewrite
\be\label{eq:option1_99s}
S_{\{i,j\},l} = rS_{\{i,j\},l-1},
\ee
where we have taken into account that the HS distance for any empty subset is zero by default, and $r$ is
taken from Eq.~\eqref{eq:rule_1}.

In the second scenario, the qubits $i$ and $j$ are each subjected to different gates, coupled to qubits $i'$ and $j'$, respectively.
First, consider the action of the gate between qubits $i$ and $i'$. For this, we introduce the ``midlayer" values of HS distance $S'_{\Omega,l}$ for the state of the qubits before the gate between qubits $i$ and $i'$ has been applied. Using Eq.~\eqref{eq:cross_gate_expr}, we get
\be\label{eqs:step1_0adc}
S_{\{i,j\},l} = \frac{1}{2}(\beta+\mu)S_{\{j\},l}+\frac{1}{2}(\beta-\mu)S'_{\{i,i',j\},l},
\ee
where
\be
\beta = \frac{q}{q^2 + 1} (A_{ii'} + B_{i'i}), \qquad \mu = \frac{q}{q^2 - 1} (B_{i'i} - A_{ii'}).
\ee
Next, we take into account the effect of gate between $j$ and $j'$. It leads us to
\be\label{eqs:step2_0adc}
\begin{split}
&S'_{\{j\},l} = \frac{1}{2}(\beta-\mu)S_{\{j,j'\},l-1}\\
&S'_{\{i,i',j\},l} = \frac{1}{2}(\beta+\mu)S_{\{i,i'\},l-1}+\frac{1}{2}(\beta-\mu)S_{\{i,i',j,j'\},l-1}
\end{split}
\ee
where 
\be
\beta' = \frac{q}{q^2 + 1} (A + B), \qquad \mu' = \frac{q}{q^2 - 1} (B - A).
\ee
Combining the expressions in Eq.~\eqref{eqs:step1_0adc} and \eqref{eqs:step2_0adc}, we get
\be\label{eq:option2_99s}
\begin{split}
S_{\{i,j\},l} &= \frac{1}{4}(\beta^2-\mu^2)\Bigl(S_{\{i,i'\},l}+S_{\{j,j'\},l-1}\Bigl)+\frac{1}{4}(\beta-\mu)^2S_{\{i,i',j,j'\},l-1}\\
&\geq \frac{1}{4}(\beta^2-\mu^2)\left(S_{\{i,i'\},l-1}+S_{\{j,j'\},l-1}\right)
\end{split}
\ee
where we used the fact that HS distance is non-negative. We observe that after applying either Eq.~\eqref{eq:option1_99s} or Eq.~\eqref{eq:option2_99s}, depending on the specific case, to the HS of all qubit pairs in $C_{l+1}$, at least one pair will be expressed in terms of the most recent HS from $C_l$. Thus,
\be
S^{(2)}_{l} \geq e^{-\Gamma} S^{(2)}_{l-1}, \qquad \Gamma = -\log\left(\min\left\{r,\frac{1}{4}(\beta^2-\mu^2)\right\}\right).
\ee
This expression allows us to derive the full solution
\be
 \max_{i\neq j} S_{\{i,j\},D} \geq e^{-\Gamma} S^{(2)}_{D-1} \geq S^{(2)}_{0} e^{-\Gamma D} =  e^{-\Gamma D},
\ee
where we have taken into account that for two distinct bitstrings $S^{(2)}_{0} = 1$.
Consider qubit indices $\{a,b\}$ such that $S_{\{a,b\},D} = \max_{i\neq j}S_{\{i,j\},D}$. Then, using the inequality in Eq.~\eqref{TS_ineq}, we get
\be
T_{F,D}: = \frac{1}{2} \mathbb{E} \|\mathcal C(\rho - \sigma)\|_1  \geq T_{\{a,b\},D}\geq S_{\{a,b\},D} \geq e^{-\Gamma D},
\ee
which completes the proof.
\end{proof}

\subsection{Linear-depth upper bound for local circuits}
\begin{theoremS}
Consider an $n$-qubit noisy random unitary circuit consisting of $D$ parallel layers of two-qubit gates.  For any pair of input states $\rho, \sigma \in \mathsf S_{2^n}$, we have
\be
\frac 12\mathbb E_{\mathcal B} \|\mathcal C(\rho-\sigma)\|_1\leq 2^{n/2}e^{-\Gamma' D},
\ee
where $\Gamma'>0$ is a constant, $\mathcal{C}$ is a noisy random circuit channel, and $\mathbb{E}$ denotes the expectation over the random two-qubit unitaries.
\end{theoremS}

\begin{proof}  Let us introduce a new set of parameters $\{S'_{G,k}\}$ in the following way. First we set that if the $d$-th gate applies within set $G$, i.e. $\Omega\subset G$, but $\Omega\neq G$, we have
\be\label{eq:s_prime_1}
S'_{G,k} = S_{G,k} = \alpha\Bigl[\left(1-\frac{r}\alpha\right) S_{G\setminus\Omega,k-1}+\frac r\alpha S_{G,k-1}\Bigl],
\ee
 where $\alpha$ is defined in Eq.~\eqref{eq:rule_1}. The exception to this is when $G=\Omega$, then
\be\label{eq:exception}
S'_{\Omega,k} = \alpha S_{\Omega,k-1}.
\ee
In this case, $S_{\Omega,k} = rS_{\Omega,k-1}\leq S'_{\Omega,k}$ because $r\leq\alpha$. Next, if $\Omega$ and $G$ overlap but $\Omega\setminus G \neq \{\emptyset\}$, we get
\be
S'_{G,k} = S_{G,k} = \beta\Bigl[\left(1+\frac{\mu}{2\beta}\right) S_{G\setminus\Omega,k-1}+\left(1-\frac{\mu}{2\beta}\right)S_{G,k-1}\Bigl],
\ee
Finally, if the gate applies outside $G$, i.e. $G\cap \Omega = \{\emptyset\}$, we get
\be\label{eq:s_prime_4}
S'_{G,k} =S_{G,k} = S_{G,k-1}.
\ee
Summarizing, the quantities $S'_{G,k}=S_{G,k}$ for all $G$ except when $G=\Omega$, in which case they are given by Eq.~\eqref{eq:exception}. Therefore, we set that $S'_{G,k}\leq S_{G,k}$. 

Using Eqs.~\eqref{eq:s_prime_1}-\eqref{eq:s_prime_4}, we get that there exist integers $N_1$ and $N_2$ and coefficients $1\geq p(Q,G)\geq 0$ corresponding to each subset $Q$, $G$ in $F$ excluding empty set (denoted as $\mathcal F\setminus\emptyset$) and satisfying $\sum_{G\in\mathcal F\setminus\emptyset}p(Q,G) = 1$ such that
\be
\forall G'\in \mathcal F\setminus\emptyset: \qquad S'_{Q,k} = \alpha^{N_1}\beta^{N_2}\sum_{G\in\mathcal F\setminus\emptyset}p(Q,G)S_{G,k-ND/2}.
\ee
Next, because every qudit participates ate least at each gate and because $Q$ cannot be empty, we have
\be
2\leq N_1+N_2\leq ND/2.
\ee
Then, we can transform
\be
S_{F,k}\leq S'_{F,k} \leq e^{-2\gamma D}\sum_{G\in\mathcal F\setminus\emptyset}p(Q,G)S_{G,0} \leq e^{-2\gamma D}\sum_{G\in\mathcal F\setminus\emptyset}p(Q,G) \leq e^{-2\Gamma' D}.
\ee
where $\gamma = -\log \max(\alpha,\beta)$ as in the condition of the Theorem. Using Eq.~\eqref{TS_ineq}, we then get
\be
T_{F,k}\leq \sqrt{N_F S_{F,k}} = 2^{n/2}e^{-\Gamma' D}~,
\ee
which completes the proof.
\end{proof}

\subsection{Constant-depth upper bound for all-to-all random circuits}
\label{sec:const_depth_rc}

Consider a model where each layer consists of a global Haar random $n$-qubit gate, followed by a layer of single-qubit noise channels. We then consider the operator
\be
X_d := \mathcal C_d(\rho-\sigma),
\ee
where we assume that $\rho$ and $\sigma$ are two orthogonal pure states. Here, $\mathcal{C}_d$ represents the noisy unitary circuit channel,
\be\label{eqs:all_to_all_channel}
\mathcal C_d := \mathcal N^{\otimes n}\circ \mathcal U_d\cdots \mathcal N^{\otimes n}\circ \mathcal U_1,
\ee
where $\mathcal U_k$ is the unitary transformation through the $k$th Haar random unitary gate, and $\mathcal N$ is the single-qubit noise channel,
\be
\mathcal N(\rho) = \sum_{i=1}^q K_i\rho K_i^\dag,
\ee
which results in the total channel
\be
\mathcal N^{\otimes n} (\rho) := \sum_{{\bf s}\in C_n} \hat K_{\bf s}\rho \hat K^\dag_{\bf s}, \qquad \hat K_{\bf s} = K_{s_1}\circ K_{s_2}\circ \dots \circ K_{s_n},
\ee
where $K_1, \dots, K_q$ are the Kraus operators for a single-qubit channel, and $C_n$ is the set of all combinations of $n$ numbers taking values $s_i \in \{1, \dots, q\}$. Our goal is to find the expectation of the trace distance between the initially orthogonal states $\rho$ and $\sigma$, i.e.
\be\label{eqs:avg_trace_dis}
\frac 12\mathbb E_{\mathcal U}\|X_d\|_1 \equiv \frac 12\mathbb E_{\mathcal U} \Tr \sqrt{X_d^2},
\ee
where $\mathbb E_{\mathcal U}$ is the expectation over Haar random gates. Then, we prove that for the replacement channel (see Definition~\ref{dfn:rep_channel}) the following result holds:

\begin{thrm}[Restatement]
Consider an $n$-qubit all-to-all random circuit of depth $D > 1$ subject to noise generated by a replacement channel with the rate parameter $0 < \gamma \leq 1$. Then
\be
\frac{1}{2} \mathbb{E} \|\mathcal{C}(\rho - \sigma)\|_1 \leq O\left(2^{n/2} \left(1 - \frac{\gamma}{2}\right)^{n(D - 1)/2}\right),
\ee
where $\mathcal{C}$ denotes the channel corresponding to the noisy random circuit, and $\mathbb{E}$ denotes the expectation over the gate unitaries.
\end{thrm}

\begin{proof} To obtain an upper bound on this expectation, we use the invariance of the trace distance with respect to the unitary transformation and the concavity of the square root operation, which gives us the inequality
\be
\mathbb E \Tr \sqrt{X_d^2} =  \mathbb E \Tr \sqrt{\mathcal U^\dag [X_d^2]} \leq \frac 12\mathbb  \Tr \sqrt{ \mathbb E\, \mathcal U^\dag [X_d^2]} = \mathbb  \Tr \sqrt{ \mathbb E_1\,\mathcal U_1^\dag\dots \mathbb E_d\,\mathcal U_d^\dag [X_d^2]},
\ee
where $\mathcal U := \prod_k \mathcal U_k$ is the noiseless circuit representation. Next, for any $X$ that satisfies $\Tr X = 0$, we get (see Eq.~\eqref{eqs:weing_calc_09} in Appendix)
\be
\mathbb E_k\, \mathcal U_k^\dag\bigl [(\mathcal N \circ \mathcal U_k [X])^2\bigl] = \mathbb E\sum_{ij} U_k^\dag (\hat K_i U_k X U_k^\dag \hat K^\dag_i)(\hat K_j U_k X U_k^\dag \hat K_j^\dag) U_k = \alpha_n X^2 +\gamma_n\Tr (X^2)I,
\ee
where $I$ is the $n$ qubit identity operator, $\alpha_n$ and $\gamma_n$ are described by Eq.~\eqref{eq:alpha_expression} and \eqref{eq:gamma_expression} in Appendix respectively. Also, we get (see Eq.~\eqref{eqs:weing_calc_10})
\be
\mathbb E \Tr (\mathcal N \circ \mathcal U_k [X])^2 = \mathbb E_{U} \sum_{ij}\Tr \Bigl[ \hat K_j^\dag \hat K_i UX U^\dag \hat K^\dag_i\hat K_j U X U^\dag  \Bigl] = \delta_n\Tr(X^2),
\ee
where the expression for $\delta_n$ is shown in  Eq.~\eqref{eq:delta_expression} in Appendix. Based on this algebra, we get
\be
\mathbb E_k\, \mathcal U^\dag_{k}[X^2_{k}] = \alpha_n X^2_{k-1} + \gamma_n \Tr(X^2_{k-1}) I.
\ee
Using these formulas self-consistently, we derive that
\be
\begin{split}
\|\mathcal C_d(\rho-\sigma)\|_1 \leq \frac 12\sqrt{ AX_0^2+B\Tr(X_0^2) I}, \qquad X_0 \equiv \rho-\sigma,
\end{split}
\ee
 the coefficients $A$ and $B$ can be found as a solution of
\be\label{eqs:matrxi_eq}
\begin{pmatrix}
A \\
B
\end{pmatrix} = \mathcal T^d \begin{pmatrix}
1 \\
0
\end{pmatrix}, \qquad \mathcal T = \begin{pmatrix}
\alpha & 0 \\
\gamma & \delta
\end{pmatrix},
\ee
where $\mathcal T$ is a simple $2\times 2$ transfer matrix. The solution of Eq.~\eqref{eqs:matrxi_eq} is
\be
A =  \alpha_n^d, \qquad B = \gamma \frac{\alpha_n^d - \delta_n^d}{\alpha_n - \delta_n} = O\left(\gamma \max(\alpha,\delta)^{d-1}\right),
\ee
where in the last expression we take an asymptotic limit of large $d$. Using the fact the initial states are pure and orthogonal, we get
\be\label{eqs:final_bounN_alltoall}
\begin{split}
\Delta(d) \leq \sqrt{A+2B}+(N_H/2-1)\sqrt{2B},
\end{split}
\ee
where $N_H = 2^n$ is the dimension of the Hilbert space. 

Using the estimates of $\alpha_n$ and $\gamma_n$ provided in Eq.~\eqref{eq:replacement_architecture_ag1} of the Appendix, we get that $\alpha_n = O((1-7\kappa/8)^n)$, $\gamma_n = O(2^{-n})$, and $\delta_n = O((1-\gamma/2)^n)$. In this limit, $\alpha_n/\delta_n \to 0$ as $n\to\infty$, therefore for large enough $n$ we have $\max(\alpha_n,\delta_n) = \delta_n$. This means
\be
\Delta(d) = O\left(\alpha_n^{d/2}\right) + O\left(N_H\gamma^{1/2}\delta_n^{d/2}\right) = O\left(2^{n/2}\left(1-\frac \gamma2\right)^{n(d-1)/2}\right)
\ee
This expression conlcudes our proof.
\end{proof}

The direct corollary of the theorem is that for $d>N_0 = 1-1/\log_2(1-\gamma/2)$ the distance decays exponentially in the total circuit volume $V = nd$. Thus all-to-all random circuits posses only constant memory. 

\section{Appendix}

\subsection{Quantum channels}

Let $\mathcal S_n = \{\rho: \mathbb C^{2^n}\to \mathbb C^{2^n} \mid \Tr\rho = 1, \rho \geq 0\}$. Below we provide definitions of replacement and generalized-damping channels.

\begin{dfn}[\textbf{Replacement Channel}]  \label{dfn:rep_channel}
For any $0 < \gamma \leq 1$ and a fixed state $\sigma^* \in \mathcal S_2$, we define the replacement channel $\mathcal N_{\rm repl}$ such that for any state $\rho \in \mathcal S_2$, the action of the channel is given by
\be
\mathcal N_{\rm rep}(\rho) := (1-\gamma)\rho +  \gamma \sigma^*.
\ee
\end{dfn}

In this formulation, the channel mixes the input state $\rho$ with the fixed state $\sigma^*$, where $\gamma$ controls the degree of replacement. Consider the case when the state $\sigma^*$ has purity $\eta = \sqrt{2\Tr((\sigma^*)^2)-1}$ this channel can be written in the form
\be
\mathcal N_{\rm rep}(\rho) = \sum_{\mu=0}^4  K_\mu \rho K^\dag_\mu,
\ee
where $K_\mu$ are the following Kraus operators
\be
\begin{split}
&K_0 = \sqrt{1-\gamma}\, I, \quad K_1 = \sqrt{{\gamma\frac{1+\eta}{2}}}|\psi_0\>\<\psi_0|,\quad K_2 = \sqrt{{\gamma\frac{1+\eta}{2}}}|\psi_0\>\<\psi_1|,\\
&K_3 = \sqrt{{\gamma\frac{1-\eta}{2}}}|\psi_1\>\<\psi_0|, \quad K_4 = \sqrt{{\gamma\frac{1-\eta}{2}}}|\psi_1\>\<\psi_1|,
\end{split}
\ee
and $\{|\psi_\mu\>\}$ are eigenstates of $\sigma^*$. The case $\eta = 0$ consititutes the case of depolarizing channel. 

\begin{dfn}[\textbf{Generalized Damping Channel}] For any $0 \leq \eta \leq 1$ and $0 < \gamma \leq 1$, we define the (generalized) damping channel as
\be\label{eq:noise_channel}
\mathcal N_{\rm damp}(\rho) := \sum_{\mu=1}^4  K_\mu \rho K^\dag_\mu,
\ee
where

\be
\begin{split}
&K_1 = 
\sqrt{\frac{1+\eta}{2}}\begin{pmatrix}
1 & 0\\
0 & \sqrt{1-\gamma}
\end{pmatrix},\quad 
K_2 = 
\sqrt{\frac{1-\eta}{2}}\begin{pmatrix}
\sqrt{1-\gamma} & 0\\
0 & 1
\end{pmatrix},\\
&K_3 = 
\sqrt{\frac{1+\eta}{2}}\begin{pmatrix}
0 & \sqrt{\gamma}\\
0 & 0
\end{pmatrix},\quad
\quad
K_4 = 
\sqrt{\frac{1-\eta}{2}}\begin{pmatrix}
0 & 0\\
\sqrt{\gamma} & 0
\end{pmatrix}.
\end{split}
\ee
\label{dfn:gen_damping}
\end{dfn}
The parameters are chosen such that this channel also has the purity of the steady state $\sigma^*$ equal to $\eta$. The case $\eta = 1$ corresponds to the commonly used notion of the amplitude damping channel.

\subsection{Weingarten calculus}
\label{sec:wg_calc}

Consider $\mathcal N$ to be a single-qubit noise channel, represented as
\be
\mathcal N(\rho) = \sum_{i=1}^q K_i\rho K_i^\dag,
\ee
Also, consider its action on $m$ qubits forming a subset $\Omega$ of the full set,
\be
\mathcal N^{\otimes m} (\rho) := \sum_{{\bf s}\in C_n} \hat K_{\bf s}\rho \hat K^\dag_{\bf s}, \qquad \hat K_{\bf s} = K_{s_1}\circ K_{s_2}\circ \dots \circ K_{s_m},
\ee
where $K_1, \dots, K_q$ are the Kraus operators for a single-qubit channel, and $C_m$ is the set of all combinations of $n$ numbers taking values $s_i \in \{1, \dots, q\}$.

Using the Weingarten calculus, we derive that
\be\label{eqs:weing_calc_09}
 E_{U}\sum_{ij} U^\dag \hat K_i UX U^\dag \hat K^\dag_i \hat K_j U X U^\dag \hat K_j^\dag U = \alpha_m X^2 + \beta_m (\Tr_\Omega X)^2\otimes I_\Omega+\omega_m\Tr_\Omega (X^2)\otimes I_\Omega,
\ee
where $\mathbb E_U$ is the expectation over the Haar distribution of $m$-qubit unitaries $U$, and $I_\Omega$ is the identity operator acting on the subset $\Omega$ that $U$ acts on, $\Tr_\Omega$ is the partial trace over the subset $\Omega$, and the coefficients are expressed by the Kraus operators of the noise channel as

\be\label{eq:alpha_expression}
\begin{split}
\alpha_m =  \frac{1}{d (d^4 - 5 d^2 + 4)}&\sum_{ij} d^2 \bigl[\Tr(K_i)\Tr(K_i^\dag K_j) \Tr(K_j^\dag)\bigl]^m +  4 \bigl[\Tr(K_i^\dag K_j K_j^\dag K_i)]^m + d^2[\Tr(K_j^\dag K_i^\dag K_j K_i)]^m  \\
&-2d \left[[\Tr(K_i)\Tr(K_i^\dag K_jK_j^\dag)]^m+  [\Tr(K_iK_i^\dag K_j )\Tr(K_j^\dag)]^m + [\Tr(K_i^\dag K_j) \Tr(K_j^\dag K_i)]^m  \right],
\end{split}
\ee

\be
\begin{split}
\beta_m = \frac{1}{d (d^4 - 5 d^2 + 4 )} \sum_{ij} & 2 \, [\mathrm{Tr}(K_i) \mathrm{Tr}(K_j^\dag) \, \mathrm{Tr}(K_i^\dag K_j)]^m + (d^2-2) \, [\mathrm{Tr}(K_i^\dag K_j K_j^\dag K_i)]^m + 2 \, [\mathrm{Tr}(K_j^\dag K_i^\dag K_j K_i)]^m\\
&-d \left[ [\mathrm{Tr}(K_i^\dag K_j) \, \mathrm{Tr}(K_j^\dag K_i)]^m + \, [\mathrm{Tr}(K_j^\dag) \, \mathrm{Tr}(K_i^\dag K_j K_i)]^m + \, [\mathrm{Tr}(K_i) \mathrm{Tr}(K_j^\dag K_i^\dag K_j) ]^m\right] 
\end{split}
\ee

\be\label{eq:gamma_expression}
\begin{split}
\omega_m = \frac{1}{d (d^4 - 5 d^2 + 4)} & \sum_{ij} 2  [\mathrm{Tr}(K_i) \mathrm{Tr}(K_j^\dag K_i^\dag K_j)]^m 
+ 2[\mathrm{Tr}(K_i^\dag K_j K_i) \mathrm{Tr}(K_j^\dag)]^m + \left(d^2 - 2 \right) [\mathrm{Tr}(K_i^\dag K_j) \mathrm{Tr}(K_j^\dag K_i)]^m \\
& -d \left[ [\mathrm{Tr}(K_i) \mathrm{Tr}(K_i^\dag K_j) \mathrm{Tr}(K_j^\dag)]^m + [\mathrm{Tr}(K_i K_i^\dag K_j K_j^\dag)]^m + [\mathrm{Tr}(K_i K_j^\dag K_i^\dag K_j)]^m \right].
\end{split}
\ee
where $m$ is the number of qubits participating in the gate and $d := 2^m$. Also, we derive that for operators that satisfy $\Tr_\Omega X = 0$, we get
\be\label{eqs:weing_calc_10}
 \mathbb E_U  \sum_{ij}\Tr \Bigl[ \hat K_j^\dag \hat K_i UX U^\dag \hat K^\dag_i\hat K_j U X U^\dag  \Bigl] = \delta\Tr(X^2),
\ee
where
\be\label{eq:delta_expression}
\delta_m = \frac 1 {d(d^2-1)} \sum_{ij}{d[\Tr(K_i^\dag K_j)\Tr(K_j^\dag K_i)]^m - [\Tr K_i K_i^\dag K_jK_j^\dag]^m}
\ee
is a real parameter.

\subsubsection*{Replacement channel}

For a single-qubit replacement channel we have
\be
\begin{split}
&\alpha_m = \frac{1}{N_m} \Bigl( N_m^2 E_1^m + 4 E_2^m + N_m^2 E_3^m - 2N_m \left( E_4^m + E_5^m + E_6^m \right) \Bigl),\\
&\omega_m = \frac{1}{N_m} \left( 2 E_4^m + 2 E_5^m + (N_m^2 - 2) E_6^m - N_m \left( E_1^m + E_2^m + E_3^m \right) \right)
\end{split}
\ee
where $N_m = 2^m$, $N_m = N_m (N_m^4 - 5 N_m^2 + 4)$, and we used the following notations
\be
\begin{split}
&E_1 : = \sum_{ij}\Tr(K_i)\Tr(K_i^\dag K_j) \Tr(K_j^\dag) =  8 - 12 \kappa + \frac{1}{2}(9 + \epsilon^2) \kappa^2,\\
&E_2 : = \sum_{ij}\Tr( K_i K_i^\dag K_j K_j^\dag ) = 2 \left(1 + \epsilon^2 \kappa^2 \right) ,
\\
&E_3 : = \sum_{ij}\Tr(K_j^\dag K_i^\dag K_j K_i) = 2 + \frac{1}{2} (\epsilon^2-3) \kappa^2,
\\
&E_4 : = \sum_{ij}\Tr(K_i)\Tr(K_i^\dag K_jK_j^\dag) = 4 - 3 \kappa + \epsilon^2 \kappa^2,
\\
&E_5 : = \sum_{ij}\Tr(K_iK_i^\dag K_j )\Tr(K_j^\dag) = 4 - 3 \kappa + \epsilon^2 \kappa^2,\\
&E_6 : = \sum_{ij}\Tr(K_i^\dag K_j) \Tr(K_j^\dag K_i)  = 4 - 6 \kappa + (3 + \epsilon^2) \kappa^2.
\end{split}
\ee
To derive the expression in the asymptotic limit of large $m$, we first note that the normalization factor behaves as $N_m = \Theta(2^{5m})$. At the same time, we can bound each component of the larger formula as
\be
\begin{split}
&\frac 12 \leq E_1\leq 8\left(1-\frac 78 \kappa\right),\qquad  2\leq E_2\leq 4 \qquad \frac 12 \leq  E_3\leq 2, \qquad 1\leq E_4,E_5\leq 4\qquad 1\leq E_6\leq 4 \left(1-\frac 12\kappa\right).
\end{split}
\ee
Most of terms in the expressions above become marginal, which leaves us with
\be\label{eq:replacement_architecture_ag1}
\begin{split}
\alpha_m = O\Biggl(\left(1-\frac 78\kappa\right)^m\Biggl) + O(2^{-2m}), \qquad \gamma_m = O(2^{-m}), \qquad \delta_m = O\Biggl(\left(1-\frac 12\kappa\right)^m \Biggl)
\end{split}
\ee

\subsubsection*{Generalized damping channel}

\be
\begin{split}
&E_1 := \sum_{ij}\Tr(K_i)\Tr(K_i^\dag K_j) \Tr(K_j^\dag) = -2 \left(\sqrt{1-\kappa }+2\right) \kappa +4 \left(\sqrt{1-\kappa }+1\right)+\frac{1}{2}
   \kappa ^2 \left(\epsilon ^2+1\right)\\
&E_2 := \sum_{ij}\Tr(K_i^\dag K_j K_j^\dag K_i) = 2 \left(1+\kappa ^2 \epsilon ^2\right)
\\
&E_3 := \sum_{ij}\Tr(K_j^\dag K_i^\dag K_j K_i) = 2 \left(\sqrt{1-\kappa }-1\right) \kappa +\frac{1}{2} \kappa ^2 \left(\epsilon
   ^2+1\right)+2
\\
&E_4 := \sum_{ij}\Tr(K_i)\Tr(K_i^\dag K_jK_j^\dag) = -\kappa +2 \sqrt{1-\kappa }+\kappa ^2 \epsilon ^2+2
\\
&E_5 := \sum_{ij}\Tr(K_iK_i^\dag K_j )\Tr(K_j^\dag) = -\kappa +2 \sqrt{1-\kappa }+\kappa ^2 \epsilon ^2+2\\
&E_6 := \sum_{ij}\Tr(K_i^\dag K_j) \Tr(K_j^\dag K_i)  = -4 \kappa +\kappa ^2 \left(\epsilon ^2+1\right)+4
\end{split}
\ee

\be
\begin{split}
&\alpha = \frac{1}{45} (\kappa -1) \left(-10 \kappa ^2+3 \left(4 \sqrt{1-\kappa }+9\right) \kappa
   -3 \left(8 \sqrt{1-\kappa }+7\right)\right),\\
   &\gamma = \frac{1}{180} (\kappa -1) \left(-12 \left(\sqrt{1-\kappa }-2\right) \kappa +24
   \left(\sqrt{1-\kappa }-1\right)+\kappa ^3 \left(6 \epsilon ^2+3\right)-\kappa ^2
   \left(18 \epsilon ^2+11\right)\right)
\end{split}
\ee

\be
\begin{split}
&1\leq E_1\leq 8(1-c\kappa), \qquad 2\leq E_2\leq 4 \qquad \frac 12 \leq  E_3\leq 2, \qquad 1\leq E_4,E_5\leq 4 \qquad 1\leq E_6\leq 4(1-\kappa/2)
\end{split}
\ee
where $c \approx 0.8452>5/6$.

In the asymptotic limit of large $m$ we get
\be\label{eq:replacement_architecture_ag_ampdad}
\begin{split}
\alpha = O\Bigl((1-c_{\rm damp}\kappa)^m + 2^{-2m}\Bigl), \qquad \gamma = O(2^{-m}).
\end{split}
\ee

\bibliographystyle{ieeetr}
\bibliography{references}

\end{document}